\documentclass[11pt]{amsart}
\usepackage{amsfonts}
\usepackage{booktabs}
\usepackage{mathrsfs}
\usepackage{subfigure}
\allowdisplaybreaks
\usepackage{hyperref}
\usepackage{color}
\usepackage{graphicx}
\usepackage{amsthm}

\usepackage{indentfirst}
\date{\today}
\makeatletter
\@addtoreset{equation}{section}

\newtheorem{theorem}{Theorem}[section]
\newtheorem{lemma}[theorem]{Lemma}

\newtheorem{remark}[theorem]{Remark}

\numberwithin{equation}{section}

\begin{document}
\title[Invasion dynamics of super invaders]{Invasion dynamics of super invaders:  Elimination of Allee effects by a strategy at the range boundary }
 \author[Y. Du, L. Li, W. Ni and N. Shabgard]{Yihong Du$^{\dag}$, Ling Li$^\ddag$, Wenjie Ni$^{\dag}$, Narges Shabgard$^{\dag}$}
 \thanks{
 \mbox{\!$^{\dag}$} School of Science and Technology, University of New England, Armidale, NSW 2351, Australia.\\
\mbox{\ \ \ \ \ \ \   \bf  Emails:} \ {\tiny \bf ydu@une.edu.au (Y. Du),\ wni2@une.edu.au (W. Ni),\  nshabgar@myune.edu.au (N. Shabgard)}.\\
 \mbox{\ \ \ \  $^{\ddag}$} College of Science, Nanjing Agricultural University,
 Nanjing, China. {\bf Email:} {\tiny \bf liling.njnu@outlook.com (L. Li)}.
}
  
\date{\today}

\begin{abstract}
  Using a reaction-diffusion model with free boundaries in one space dimension for a single population species with  density $u(t,x)$ and population range $[g(t), h(t)]$, we demonstrate that the Allee effects can be eliminated if the species maintains its population density at a suitable level at the range boundary by advancing or retreating the fronts. It is proved that with such a strategy at the range edge the species can invade the environment successfully with all admissible initial populations, exhibiting the dynamics of  super invaders. Numerical simulations are used to help understand what happens if the population density level at the range boundary is maintained at other levels. If the invading cane toads in Australia used this strategy at the range boundary to become a super invader, then our results may explain why  toads near the invading front evolve to have longer legs and run faster.

\vspace{0.6cm}

\noindent
\textbf{Keywords}: Invasion dynamics, Allee effects, reaction-diffusion equation, free boundary

\medskip

\noindent
\textbf{AMS Subject Classification}: 35B40, 35K55, 35R35
\end{abstract}
\maketitle

\section{Introduction}
Understanding the dynamical behaviour of invasive species is a central problem in invasion biology, and reaction-diffusion equations have proven to be a useful tool for this purpose \cite{SK, CC}. In the existing reaction-diffusion models for propagation, the long-time dynamics is primarily determined by the reaction terms (or growth terms in the context of population dynamics), and three types of reaction terms have been widely used, namely monostable, bistable and combustion types of reactions. The last two types of reaction terms can  capture the so called Allee effects in population dynamics, which assert that small population size/density may cause inbreeding depression leading to negative growth and eventual vanishing of a species. 

However, our analysis of a reaction-diffusion model with  free boundaries in this paper suggests  that an invasive species may eliminate the Allee effects and spread successfully by a  strategy at its range boundary alone. More precisely, Allee effects can be avoided if members of the species at the edge of the population range maintain a favourable density there by advancing or retreating the fronts, and this strategy  at the range edge will  lead to consistent successful invasion. This appears to be a new phenomenon in reaction diffusion models, and we hope it may shed some light on the understanding of the mechanisms of biological invasion. 

Indeed, the assumption here  means that for members of the species near the range boundary, their movement in space is governed by two factors: (a) random dispersal similar to all other members, and (b) keeping the density at the front at a preferred level. These two factors are balanced at the range boundary by the species moving the front forward or backward, which  is the main reason for the super invasion dynamics.
 If the invading cane toads in Australia used this strategy at the range boundary to become a super invader, then our results may explain why  toads near the invading front evolve to have longer legs and run faster \cite{P-Nature}: advancing and retreating faster makes it easier to maintain the favoured population level at the front.
\medskip

To put this research into perspective, let us briefly recall some background for the modelling of species spreading.

\subsection{Fisher-KPP type reaction-diffusion models}
 Starting from the pioneering works of Fisher \cite{Fisher} and  Kolmogorov, Petrovsky and Piskunov (KPP) \cite{KPP},
  the spreading behaviour of an invading or new species has been widely modelled by the
  Cauchy problem
\begin{equation}
\label{Cauchy}
\left\{
\begin{aligned}
&U_t=d\Delta U+f(U) && \mbox{ for } x\in\mathbb R,\; t>0,\\
&U(0,x)=U_0(x) && \mbox{ for } x\in \mathbb R,
\end{aligned}
\right.
\end{equation}
where $U(t,x)$ stands for the population density of the concerned species at time $t$ and spatial location $x$, with initial population density $U_0(x)$ assumed to be nonnegative with nonempty compact support, to represent the fact that initially the population exists only locally in space. Fisher \cite{Fisher} assumed $f(U)=U(1-U)$ and KPP \cite{KPP} allowed more general functions $f$ but with similar behaviour; they all belong to the more general class of monostable functions, which are functions with the following properties:
\[
{\bf (f_m):} \ \ \  f\in C^1,\ f(0)=f(1)=0,\ f'(0)>0>f'(1),\; (1-u)f(u)>0 \mbox{ for } u\in (0,1)\cup(1,\infty).
\]
\medskip
A striking feature of \eqref{Cauchy} with $f$ satisfying ${\bf (f_m)}$ is that it predicts consistent successful spreading with an asymptotic spreading speed: There exists $c^*>0$ such that  for any small $\epsilon>0$, 
 \begin{equation}\label{c*}
  \lim_{t\to\infty,|x|<(c^*-\epsilon)t}U(t,x)=1,~~~ \lim_{t\to\infty,|x|\geq(c^*+\epsilon)t}U(t,x)=0.
 \end{equation}
 This fact was  proved by Aronson and Weinberger \cite{AW,AW1978}, where $c^*$ was first determined independently by Fisher \cite{Fisher} and KPP \cite{KPP} in 1937, who found in \cite{Fisher, KPP}
 that \eqref{Cauchy} admits a traveling wave $U(t,x)=\phi(x-ct)$ with $\phi(-\infty)=1$ and $\phi(\infty)=0$  if and only if  $c\geq c^*:=2\sqrt {f'(0)d}$ and they further claimed that $c^*$ is the spreading speed of the species. The existence of a spreading speed was supported by numerous observations of real world examples, such as the spreading of muskrats in Europe in the early 20th century; see
  \cite{S51, SK} for more details.
  
  The above result of  Aronson and Weinberger  on \eqref{Cauchy} has subsequently been proved to be rather robust;  for example, the existence of an asymptotic spreading speed has been established for propagation in  various heterogeneous environments (see, e.g., \cite{BHN, BHRo, LZ,  W82, wein02}).
  \medskip
  
  \subsection{Allee effects} The growth function $f$ in \eqref{Cauchy} being monostable means, in biological terms, that the environment is favourable for the growth of the species as long as its density is between 0 and 1, with 1 standing for the normalised carrying capacity of the environment. However, it has been observed that very often small population density may cause inbreeding depression which leads to negative growth or even extinction of a species; such a phenomenon is  known as the Allee effects in the literature (\cite{AB, KBD}). These effects can be captured by \eqref{Cauchy} when $f$ is of bistable type, namely it has the following properties:
  \begin{equation*}
	(\bf{f_b})\ \ 
	\left\{
	\begin{aligned}
		& f\in C^1,\ f(0)=f(\theta)=f(1)=0\ \text{for\ some}\ \theta\in (0,1), \ \quad f'(0)<0,\quad f'(1)<0,\\
		& f(u)< 0\ \text{in} \ (0,\theta)\cup (1,\infty), \quad f(u)>0\ \text{in} \ (\theta,1),  \quad\int_{0}^{1} f(s) ds >0.
	\end{aligned}
	\right.
\end{equation*}  
  The constant $\theta$ here is known as the Allee threshold density, below which the population has negative growth. It was shown in \cite{AW, AW1978} that
  when $f$ satisfies ${\bf (f_b)}$, for small initial population $U_0(x)$ the species vanishes eventually ($U(t,x)\to 0$ as $t\to\infty$), and the population spreads successfully $(U(t,x)\to 1$ as $t\to\infty$) when $U_0(x)$ is large. The  sharp threshold results in \cite{DM} further imply that generically only these two types of long-time dynamical behaviour are possible. Moreover, as in the monostable case, when spreading is successful, there is also a spreading speed determined by the associated traveling wave problem \cite{AW, AW1978}, and the density function $U(t,x)$ for large $t$ is well approximated by the wave profile function \cite{FM}.
  
  The stationary problem of \eqref{Cauchy} with a bistable $f$ also arises in material science and known as the Allen-Cahn equation \cite{AC, FT}.
 In combustion theory, \eqref{Cauchy} is used to model the temperature change with time, and the nonlinear function $f$ in such a context  is usually assumed to be of combustion type \cite{Kanel}, characterised by the following properties:
 \begin{equation*}
	(\bf{f_c})\ \ 
	\left\{
	\begin{aligned}
		& f\in C^1, \ f(u)\equiv 0\ \text{in} \ [0,\theta]\ \text{ for\ some}\ \theta\in (0,1), \ f(1)=0> f'(1),\\
		&  f(u)>0\ \text{in} \ (\theta,1), \quad f(u)<0\ \text{in} \ (1,\infty).
	\end{aligned}
	\right.
\end{equation*} 
In this context, $\theta$ is known as the ignition temperature. Qualitatively, the long-time dynamics of \eqref{Cauchy} with $f$ of $({\bf f_c})$ type is similar to
the situation that $f$ is of bistable type $({\bf f_b})$; namely for small initial function $U_0(x)$, vanishing happens ($U(t,x)\to 0$ as $t\to\infty$), and the  spreading is successful $(U(t,x)\to 1$ as $t\to\infty$) when $U_0(x)$ is large (see, e.g., \cite{Kanel, AW, AW1978}), and the  sharp threshold results in \cite{DM} imply that generically only these two types of long-time dynamical behaviour are possible. 
 
 In biological terms, a combustion type of $f$ can be viewed as reflecting some kind of weak Allee effects, while a bistable $f$  reflecting strong Allee effects.
 
  In this paper, we propose to use the following more general class of functions to reflect Allee effects:
 \[
 ({\bf f_{A}}):\ \  \left\{
	\begin{aligned}
		&f\in C^1,\ \ f(0)=f(\theta)=f(1)=0\ \text{for\ some}\ \theta\in [0,1),  \quad f'(0)\leq 0,\quad f'(1)<0,\\
		& f(u)\leq 0\ \text{in} \ [0,\theta], \quad f(u)>0\ \text{in} \ (\theta,1), \quad f(u)<0\ \text{in} \ (1,\infty),\ \int_0^1f(s)ds>0.
	\end{aligned}
	\right.
 \]
We will henceforth call functions satisfying ${\bf (f_{A})}$ of  Allee type. 

For each $f$ of type ${\bf (f_A)}$, since $\displaystyle \int_0^u f(s)ds\leq 0$ for $u\in [0, \theta]$ and $\displaystyle \int_0^1f(s)ds>0$, it is easily seen that there exists a unique $\theta^*=\theta^*_f\in [\theta, 1)$ such that
\begin{equation}\label{theta*}
\int_0^{\theta^*}f(s)ds=0,\ \int_0^uf(s)ds>0 \mbox{ for } u\in (\theta^*, 1].
\end{equation}

Clearly ${\bf (f_{A})}$ functions include those in ${\bf (f_{b})}$, and as mentioned above, we call this type of growth functions strong Allee type. If $f$ satisfies ${\bf (f_{A})}$ but not ${\bf (f_{b})}$, we will henceforth say it is of weak Allee type, which includes in particular
functions satisfying ${\bf (f_{c})}$ or satisfying ${\bf (f_{A})}$ with $\theta=0$ (which implies  $f'(0)=0$); for these two special classes of weak Allee functions clearly $\theta^*=\theta$. In general, from \eqref{theta*} we have
\[
 \theta^*_f\in \begin{cases} [\theta, 1) &\mbox{ if $f$ is of weak Allee type},\medskip\\
 (\theta, 1) & \mbox{ if $f$ is of strong Allee type.}
\end{cases}
\]

\subsection{Reaction-diffusion models with a Stefan type free boundary}
To describe spreading in population dynamics, the models in the previous subsections have a shortcoming:  They do not give the precise location of the evolving population range, since although $U(0,x)=U_0(x)$ has compact support, for any $t>0$, one has $U(t,x)>0$ for all $x\in\mathbb R$. Therefore the population range determined by \eqref{Cauchy}
 is $\Omega(t):=\{x: U(t,x)>0\}=\mathbb R$ once $t>0$, although $\Omega(0)$ is a bounded set by assumption ($U_0(x)$ has compact support).
  To avoid this shortcoming, one may nominate a small constant $\sigma\in (0,1)$,
and regard
 \[
\Omega_\sigma(t):=\{x: U(t,x)>\sigma\}
\]
  as the population range, which is  a bounded set for all $t>0$. Then
\[
\Gamma_\sigma(t):=\{x: U(t,x)=\sigma\}
\]
 can be viewed as the  spreading front, and \eqref{c*} implies that for any small $\epsilon>0$ and all large $t>0$,
 \[
 \Gamma_\sigma(t)\subset\{x: (c^*-\epsilon)t\leq |x|\leq (c^*+\epsilon)t\}.
 \]
 Therefore  \eqref{c*} can be interpreted as saying that the fronts go to infinity with asymptotic speed $c^*$. It should be noted that $c^*$ is determined by the associated traveling wave problem, and is independent of  the choice of $\sigma\in (0,1)$ in $\Omega_\sigma$. Interestingly, however, in the real world  spreading speed  is obtained from observations of the actual range expansion; see subsection 1.4 below for more details.

\smallskip

To avoid using an artificial number $\sigma$ in the expression of the population range,  Du and Lin \cite{DL10} modified \eqref{Cauchy} into a free boundary problem
of the form
 \begin{equation}\label{free-bound-0}
  \begin{cases}
  u_t-du_{xx}=f(u), & t>0, \ g(t)<x<h(t),\\
    u(t,g(t))= u(t,h(t))=0,&t>0,\\
    g'(t)=-\mu u_x(t,g(t)), & t>0,\\
     h'(t)=-\mu u_x(t,h(t)),&t>0,\\
    u(0,x)=u_0(x),\ -g(0)=h(0)=h_0,& -h_0\leq x\leq h_0,
  \end{cases}
  \end{equation}
  where $\mu>0$ is a constant and the population range is explicitly given by the interval $[g(t), h(t)]$ in the model.  The free boundary conditions in \eqref{free-bound-0} coincide with the Stefan conditions in the classical free boundary model for melting of ice in contact with water, and for the biological setting here, they can be derived from some biological assumptions \cite{BDK}: If we assume that in order to expand the population range, the  species  sacrifices $k$ units of its population at the range boundary per unit of  time and space, then $\mu=d/k$. 
  
  For $f$ of monostable type, it has been shown (see \cite{dl2015}) that 
 \eqref{free-bound-0} exhibits a spreading-vanishing dichotomy for its long-time dynamics: As $t\to\infty$, either $[g(t), h(t)]$ converges to a finite interval $[g_\infty, h_\infty]$ and $u(t,x)\to 0$ uniformly (the vanishing case), or $[g(t), h(t)] \to \mathbb R$ and $u(t,x)\to 1$ (the spreading case). Moreover, in the latter case, there exists an asymptotic spreading speed determined by an associated traveling wave problem (called a semi-wave problem in the literature partly due to the fact that the wave profile function is only defined over the half line, partly due to the need to distinguish it from the classical traveling wave problem of Fisher and KPP).   
Furthermore, it was shown in \cite{DMZ} that when spreading occurs, the density function $u(t,x)$ is well approximated  by the semi-wave profile function as $t\to\infty$.  

When $f$ is of the type ${\bf (f_b)}$ or ${\bf (f_c)}$, it follows from \cite{dl2015} that for all small initial population $u_0(x)$ vanishing happens, and for all large $u_0(x)$ spreading happens; the sharp transition result in \cite{dl2015} implies that generically these are the only possible long-time behaviour of \eqref{free-bound-0}.
The results in \cite{DMZ} on the spreading speed and population density  also apply to these cases when spreading is successful.

\subsection{Super invaders} The dynamics displayed by the models in the previous two subsections, especially those with $f$ of bistable type ${\bf (f_b)}$
or of combustion type ${\bf (f_c)}$, where the population vanishes when its initial size is small  and spreads successfully when the initial size is large, agrees with many real world observations \cite{DBS}; for example, 
744 of 1466 (51\%) introduction events of birds in the global data set analyzed by Blackburn
\& Duncan \cite{BD} did not result in establishment.

However, there are examples of super invaders which spread successfully against all  odds. The spreading of muskrats  in Europe in the early 1900s is an example  with very small initial size: In 1905, five escaped muskrats from a farm near Prague resulted in a successful invasion of the entire European continent  within less than 3 decades (\cite{SK}). This example enabled Skellam \cite{S51}  to observe the phenomenon of  spreading with an asymptotic speed (as claimed by Fisher and KPP in 1937): He calculated the area of the
muskrat range from a map obtained from field data, took the square root (which
gives a constant multiple of the range radius) and plotted it against years, and found that the data
points lay on a straight line; see Fig, 1.

\begin{center}
\vspace{-1cm} 
\includegraphics[height=100mm]{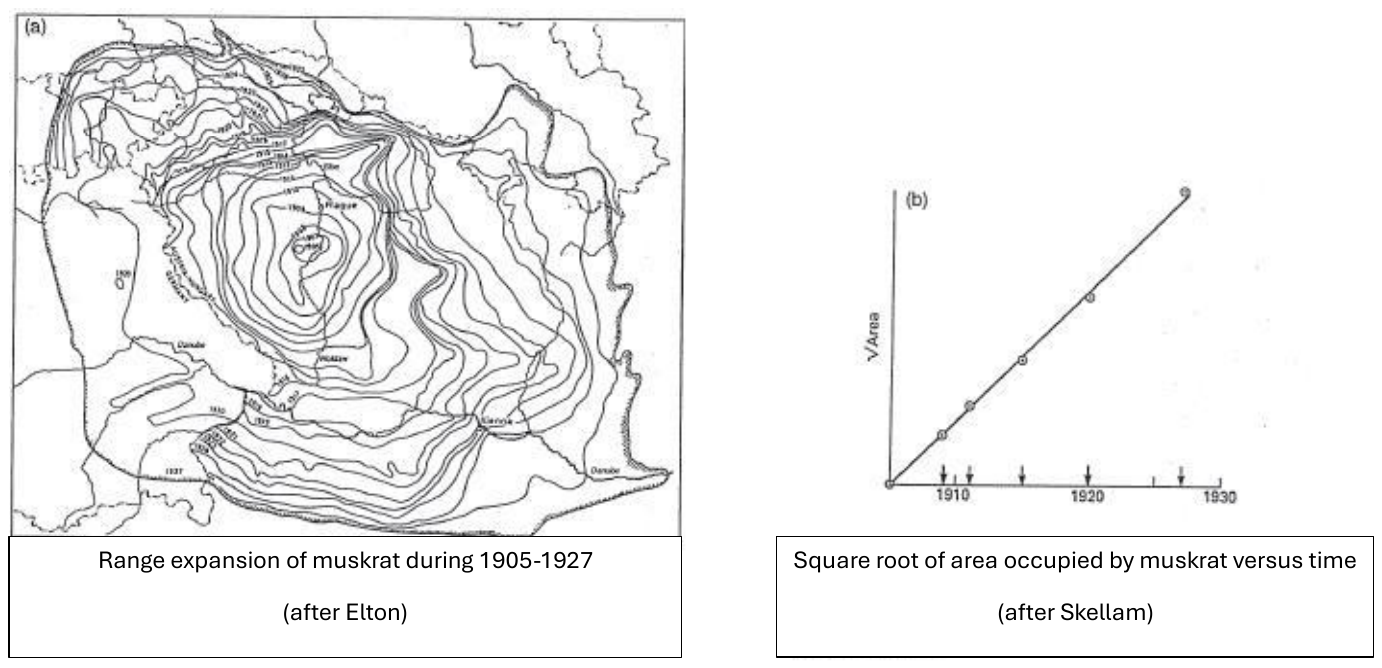}
\\
\vspace{-3cm}
{\small Fig. 1. \ Evolution of the  population range of muskrats in Europe during 1905-1927}
\end{center}

 \bigskip
 
 Another well known super invader is the cane toads in Australia.
Introduced in Australia from Hawaii in 1935, as an attempt to control ``cane beetles" in sugar cane fields of Northern Queensland, the cane toads ended up
a pest causing significant environmental detriment (and with no evidence that they have affected the number of cane beetles which they were introduced to prey upon). The invasion of cane toads in Australia is continuing, at an estimated speed of 40-50 kilometres per year; see Fig. 2.

\begin{center}
\vspace{-0.8cm}
\includegraphics[height=100mm]{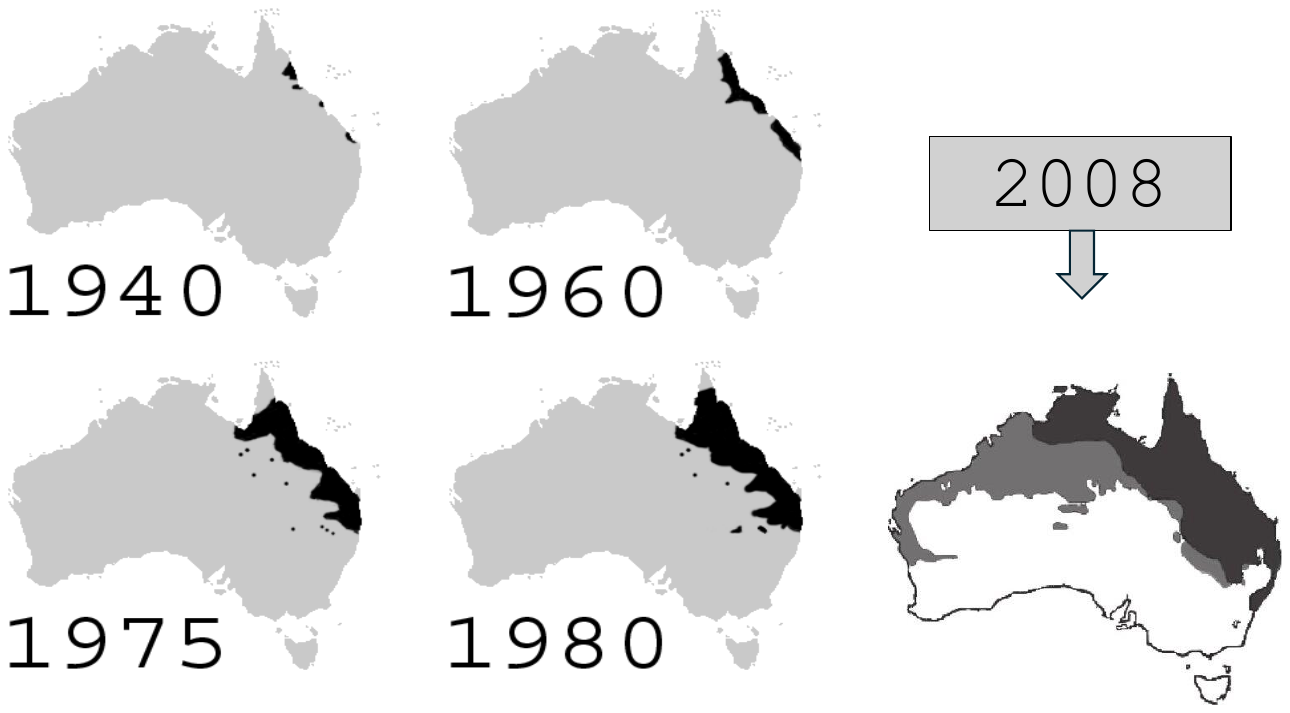}
\vspace{-2.6cm}

{\small Fig. 2. Evolution of the  population range of cane toads over the years  (graphs taken from the internet) }

\end{center}

\bigskip

Based on field observations, biologists have noticed  several distinctive behaviours of the invading cane toads, which include:
\begin{itemize} 
\item[(a)] toads at the front have longer legs and move faster \cite{P-Nature},
\item[(b)]  toads can invade  into much less favourable territory  with lower speed, including temporary front retreats \cite{M-BI}. 
\end{itemize}

\subsection{Reaction-diffusion models with traits}
 
 The observation that toads at the front have longer legs and move faster \cite{P-Nature} is an example of non-uniform space-trait distribution in an invasive population. Therefore it is natural to assume that the population density function $u$ depends on a certain trait denoted by $\theta$, apart from its dependence on time $t$ and spatial location $x$. Such an approach was taken in several recent works; see, for example, \cite{BH, BHR,  BCHK, AGHR}.

More precisely, the model in \cite{BH} has the form
\begin{equation}\label{traits-fb}
u_t=\theta u_{xx}+u_{\theta\theta}+f(u)\ \ \mbox{ for } t>0,\ x\in\mathbb{R},\ \theta\geq \bar\theta,
\end{equation}
 where $f(u)$ is a bistable function, and $\bar\theta>0$ is a constant.  So in this model the variable $\theta$ stands for a trait that represents the spatial dispersal rate of the species. 
 In \cite{BHR}, a special monostable $f(u)$ is used, and both local and nonlocal crowding effects are considered; more precisely, the following model is analysed for local crowding:
 \begin{equation}\label{traits-kpp-loc}
 u_t=\theta u_{xx}+u_{\theta\theta}+u(1-u)\ \  \mbox{ for } t>0,\ x\in\mathbb{R},\ \theta\geq \bar\theta,
 \end{equation}
 and for nonlocal crowding, \eqref{traits-kpp-loc} is modified to
 \begin{equation}\label{traits-kpp-nonloc}
 u_t=\theta u_{xx}+u_{\theta\theta}+u(1-\int_{\bar\theta}^\infty ud\theta)\ \ \mbox{ for } t>0,\ x\in\mathbb{R},\ \theta\geq \bar\theta.
 \end{equation}
 A common feature shared by \eqref{traits-fb}, \eqref{traits-kpp-loc} and \eqref{traits-kpp-nonloc} is that the species propagates in space with a superlinear rate in time, of the order $t^{3/2}$ as $t\to\infty$; see \cite{BH, BHR} for details, and also see \cite{BCHK} for a variation of \eqref{traits-kpp-loc} where the intrinsic growth rate 1 is replaced by $1-m(\theta)$. 
 
 In \cite{AGHR}, a bistable growth function was used with the trait $\theta$ representing the Allee threshold level, which varies over a finite interval $(\theta_*, \theta^*)$ in $\mathbb{R}$, and the model has the form
 \begin{equation}\label{traits-Allee}
 u_t=d u_{xx}+\alpha u_{\theta\theta}+u(\int_{\theta_*}^{\theta^*} ud\theta-\theta)(1-\int_{\theta_*}^{\theta^*} ud\theta) \ \ \mbox{ for } t>0,\ x\in\mathbb{R},\ \theta\in (\theta_*,\theta^*).
 \end{equation}
 Here $d$ and $\alpha$ are fixed positive constants. Interesting dynamical behaviour of \eqref{traits-Allee} was established in \cite{AGHR}, although many questions remain to be answered.

 \subsection{A free boundary model for super invaders} 
In \cite{d2024}, the following variation of  \eqref{free-bound-0} was considered,
\begin{equation}\label{a}
\begin{cases}
u_t-du_{xx}=f(u),  & t>0,\; g(t)<x<h(t),\\
u(t,g(t))=u(t,h(t))=\delta, & t>0,\\
g'(t)=-\frac{d}{\delta}u_{x}(t,g(t)), &  t>0,\\
h'(t)=-\frac{d}{\delta}u_x(t,h(t)), &  t>0,\\
-g(0)=h(0)=h_0, u(0,x)=u_0(x), & -h_0\leq x\leq h_0.
\end{cases}
\end{equation}
Similar to \eqref{free-bound-0}, here $u(t,x)$ stands for the population density and $[g(t), h(t)]$ represents the population range.
The initial function   $u_0$ is assumed to belong to  
$${X}(h_0):=\{\phi\in C^2([-h_0, h_0]): \phi(\pm h_0)=\delta, \ \phi>0 \ \text{in} \ [-h_0, h_0]\}.$$
While the population range in \eqref{free-bound-0} always expands as time increases, this is no longer the case for \eqref{a}, where the fronts $x=h(t)$ and $x=g(t)$ 
 may advance or retreat as time increases. The constant $\delta\in(0,1)$ in \eqref{a} represents the species' preferred density, and the equations governing the evolution of the free boundaries $x=g(t)$ and $x=h(t)$, namely the second, third and fourth equations in \eqref{a}, can be deduced from the biological assumption that the species maintains its preferred  density $\delta$ at the range boundary by advancing or retreating the fronts; see \cite{d2024} for a detailed deduction and more background. This assumption means that for members of the species near the range boundary, their movement in space is governed by two factors: (a) random movement similar to all other members, and (b) advance or retreat to keep the density at the front at the preferred level $\delta$.

When the growth function satisfies ${\bf (f_m)}$, it was shown in \cite{d2024} that the unique solution $(u(t,x), g(t), h(t))$ of  \eqref{a} exhibits successful spreading for all admissible initial data (namely for every $u_0\in X(h_0)$ with $h_0>0$). We will show in this paper that successful spreading also happens consistently when $f$ satisfies ${\bf (f_A)}$ and $\delta\in (\theta^*_f, 1)$, and so  \eqref{a} is capable of exhibiting the dynamics of super invaders.

\subsection{Main results and structure of the paper}
The main purpose of this paper is to demonstrate that successful spreading is always achieved by \eqref{a}  when $f$ is of Allee type, namely ${\bf (f_{A})}$ holds, provided that the preferred density $\delta$ at the range boundary satisfies $\delta\in (\theta^*_f, 1)$.
  In other words, the Allee effects exhibited in \eqref{Cauchy} and \eqref{free-bound-0} can be eliminated if the species chooses the strategy to advance or retreat the  boundary of the population range to keep the density there at a level $\delta\in (\theta^*_f, 1)$.

  \begin{theorem}\label{th1.3}
	Suppose that $f$ is of Allee type ${\bf (f_A)}$ and $\delta\in (\theta_f^*, 1)$. Then for every initial function $u_0 \in X(h_0)$, \eqref{a} has a unique solution $(u(t,x), g(t), h(t))$ defined for all $t>0$. Moreover,  as $t \to \infty$,
	\begin{align*}
		(g(t), h(t)) \to (-\infty, \infty), \quad u(t, x) \to 1 \quad \text{locally uniformly in } x \in \mathbb{R}.
	\end{align*}
Furthermore, there exist some constants $\hat{h}, \hat{g} \in \mathbb{R}$ such that
 \begin{align*}
 	&\lim_{t \to \infty}[h(t) - c^* t] = \hat{h}, \quad \lim_{t \to \infty} h'(t) = c^*,\\
 	&\lim_{t \to \infty}[g(t) + c^* t] = \hat{g}, \quad \lim_{t \to \infty} g'(t) = -c^*,\\
 		&\lim_{t \to \infty} \sup_{x \in [0, h(t)]} |u(t, x) - q^*( h(t) - x)| = 0,\\
 			&\lim_{t \to \infty} \sup_{x \in [g(t), 0]} |u(t, x) - q^*( x - g(t))| = 0,
 \end{align*}
where $(c^*, q^*)$ is the unique solution pair $(c,q)$ of 
		\begin{equation}\label{semi1}
			\begin{cases}
				dq'' - cq' + f(q) = 0, \quad q > 0 \quad \text{in } (0,\infty), \\
				q(0) = \delta, \quad q(\infty) = 1, \quad q'(0) = \frac{c \delta}{d},\quad q'>0 \mbox{ in } [0,\infty).
			\end{cases}
		\end{equation}
	
\end{theorem}
\bigskip

To make the reading of this paper more accessible, we have tried to avoid as much mathematical technicalities as possible before the end of Section 2, and so most of the technical contents are contained in Section 3.

To emphasise the new dynamical feature of \eqref{a}, the part of  Theorem \ref{th1.3} on the long-time dynamical behaviour of \eqref{a} is proved in Section 2, while the proof of the existence and uniqueness  result  for \eqref{a} is left to Section 3, where a more general system \eqref{c} will be treated.
   In Section 2, we also discuss, through numerical simulations, what happens to \eqref{a} if the assumption $\delta\in (\theta_f^*, 1)$ is not satisfied.  

In Section 3, we first prove the existence and uniqueness of a solution to the more general system \eqref{c} (which implies the corresponding conclusion for \eqref{a} stated in Theorem \ref{th1.3}), then we
use Theorem \ref{th1.3}  to show the elimination of Allee effects for \eqref{c} (see Theorem \ref{th3.1}).
Moreover, we prove some further theoretical results on the long-time dynamics of  \eqref{c} (see Theorems  \ref{th3.2} and \ref{th3.3}) and confirm several conjectures listed in Section 2. These imply, in particular, that in Theorem \ref{th1.3}, if the assumption on $\delta$ is changed to $\delta>1$ then vanishing always happens, and a transition phenomenon happens when $\delta=1$. While some of the techniques in this section are based on existing ones, several new techniques are introduced in the proof of Lemma \ref{lemma2.5} and  in subsections 3.5 and  3.6, which may find applications in future work.
\medskip

We refer to \cite{BDLZ, CLZ, FLWW, KMY} for a small sample of other variations of the free boundary model \eqref{free-bound-0}. See also \cite{BLM, BLM2025, LFZ} for rather different free boundary models for the evolution of range boundary. They are all for very different purposes from the one  in this paper.

\section{Longtime dynamics of \eqref{a}}

\subsection{Proof of Theorem \ref{th1.3}} In order not to interrupt thoughts with too much technicalities, and also to make the reading more accessible to biologically oriented readers, the existence and uniqueness of the solution to \eqref{a} is proved in  Section 3, where a more general problem than \eqref{a} will be considered. 

Here we only prove the longtime behavior of the solution. The strategy is to reduce the problem to the case of \eqref{a} with a monostable growth function, and then all the conclusions will follow from \cite{d2024}.

A key step is to show that when ${\bf (f_A)}$ holds and $\delta\in (\theta^*_f, 1)$,  there exists $T_0>0$ such that
 \begin{equation}\label{>delta} u(t, x) \geq \delta \mbox{  for $ t\geq T_0$ and $x\in [g(t), h(t)]$}.
 \end{equation}
 
 Once \eqref{>delta} is established,  the conclusions will follow from the monostable case in \cite{d2024} as explained below. Let 
 \[
 (\tilde u(t,x), \tilde g(t),\tilde h(t)):=(u(T_0+t, x), g(T_0+t), h(T_0+t)).
 \]
 Then clearly 
 \begin{equation}\label{b}
\begin{cases}
\tilde u_t-d \tilde u_{xx}=f(\tilde u), & t>0,\; \tilde g(t)<x< \tilde h(t),\\
\tilde u(t, \tilde g(t))=\tilde u(t,\tilde h(t))=\delta, &  t>0,\\
\tilde g'(t)=-\frac{d}{\delta} \tilde u_{x}(t, \tilde g(t)),& t>0,\\
\tilde h'(t)=-\frac{d}{\delta}\tilde u_x(t, \tilde h(t)), & t>0,\\
\tilde g(0)=g(T_0),\ \tilde h(0)=h(T_0), \  \tilde u(0,x)=u(T_0, x), & g(T_0)\leq x\leq h(T_0).
\end{cases}
\end{equation}

 We now redefine $f(s)$ for $s\in (0, \delta)$ to obtain a new $\tilde f(s)$ such that $\tilde f$ satisfies ${\bf (f_m)}$ and $\tilde f(s)=f(s)$ for $s\geq \delta$.
 Since $\tilde u\geq \delta$ by \eqref{>delta}, we see that $(\tilde u(t,x), \tilde g(t),\tilde h(t))$ satisfies \eqref{b} with $f$ replaced by $\tilde f$. Therefore we can use the results in \cite{d2024} to conclude that $(\tilde u(t,x), \tilde g(t),\tilde h(t))$ has all the properties stated in Theorems 1.2 and 1.4 there, which is equivalent to saying that $(u, g, h)$ has all the properties stated in Theorem \ref{th1.3}. Let us note that the requirement $h(0)=-g(0)$ in \eqref{a} 
 (and in \cite{d2024}) is for convenience only; we can easily recover $\tilde h(0)=-\tilde g(0)$ in \eqref{b}  by a suitable translation of the variable $x$.
 \medskip
 
 It remains to prove \eqref{>delta}. We will need a well known result (see, e.g., \cite[Theorem 4.1]{AW1978}) on traveling wave solutions of 
 \begin{equation}\label{F}
	u_t - d u_{xx} = F(u), \quad t > 0, \; x \in \mathbb{R},
	\end{equation}
with
 \( F \in C^1([0, \infty)) \) satisfying the following (unnormalised) bistable condition:
\begin{equation*}
	(\bf{F_b}):\ \ \ \ 
	\left\{
	\begin{aligned}
		& F(0) = F(P) = F(Q) = 0 \ \text{with } 0<P <Q<\infty, \quad F'(0)<0,\ F'(Q) < 0, \\
		& F < 0 \ \text{in } (0, P)\cup (Q,\infty), \quad F > 0 \ \text{in } (P, Q), \ \int_{0}^{Q} F(s) \, ds > 0.
	\end{aligned}
	\right.
\end{equation*}

\begin{lemma}\cite{AW1978}\label{lemma3.1}
	Let \( F \) satisfy \(\mathbf{(F_b)}\). Then there exists a constant \( c_* > 0 \) such that \eqref{F}
	has a traveling wave solution $u(t,x)=\phi(x+c_*t)$ with speed \( c_*>0 \); more precisely,   the  problem
	\begin{equation*}
		\begin{cases}
			d \phi''(x) - c_* \phi'(x) + F(\phi(x)) = 0, \quad x \in \mathbb{R}, \\
			\phi(-\infty) = 0, \quad \phi(+\infty) = Q,
		\end{cases}
	\end{equation*}
	admits a solution \( \phi \in C^2(\mathbb{R}) \) which is strictly increasing.
\end{lemma}

We are now ready to prove \eqref{>delta}.
	Since $\delta\in (\theta^*_f, 1)$,  we have $\int_0^\delta f(s)ds>\int_0^{\theta_f^*}f(s)ds=0$. Therefore we are able to choose a function $\hat{f} \in C^1$  sufficiently close to $f$ in $L^\infty$ such that $\hat{f}(s) \leq f(s)$ for $s \geq 0$, and  $\hat f$ satisfies $\mathbf{(F_b)}$  with $(P, Q)=(\hat\theta, \delta)$ for some $\hat\theta\in [\theta, \delta)\cap (0,\delta)$. In particular, we have $\hat{f}(0) = \hat{f}(\hat\theta) = \hat{f}(\delta)=0$ and $\hat{f}(s) > 0$ for $s \in (\hat\theta, \delta)$.  Then, by Lemma \ref{lemma3.1}, the following problem
	\begin{equation}\label{fg}
		\left\{
		\begin{aligned}
			&dq'' - cq' + \hat{f}(q) = 0, \quad z \in \mathbb{R}, \\
			&q(-\infty) = 0, \quad q(\infty) = \delta,
		\end{aligned}
		\right.
	\end{equation}
	has a solution pair $(c,q)=(c_0,q_0)$ with $c_0>0$ and $q_0(\cdot)$ strictly increasing.
	
	Next, we will make use of $q_0(z)$ to construct a lower solution to bound $u(t,x)$ from below. Since $q_0(-\infty) = 0$, we can choose $L > 0$ sufficiently large such that $q_0(h_0 - L) \leq \min_{x \in [-h_0, h_0]} u_0(x)$. Define
	\[
	\underline{u}(t, x) := \max\{q_0(ct - x - L), q_0(ct + x - L)\}.
	\]
	Due to $\hat {f} \leq f$ and  $\underline{u}_x(t, 0^-) = -q_0'(ct - L)\leq 0 \leq q_0'(ct - L) = \underline{u}_x(t, 0^+)$, it is easy to see
	 that $\underline{u}(t, x)$ satisfies (in the weak sense)
	\[
	\underline u_t \leq  d \underline u_{xx} + f(\underline u)  \quad \text{for } t > 0, \; x \in \mathbb{R}.
	\]
	Additionally, we have
	\[
		0 \leq \underline{u}(t, x) \leq  \delta \quad \text{for } x \in \mathbb{R},\ \mbox{which implies } \underline u(t, x)\leq u(t,x) \mbox{ for } t>0,\ x\in\{g(t), h(t)\},\]
		and
		\[
		\underline{u}(0, x) = \max\{q_0(-x - L), q_0(x - L)\} \leq q_0(h_0 - L) \leq u_0(x) \quad \text{for } x \in [-h_0, h_0].
	\]
	Therefore we can apply the standard comparison principle over $\{(t, x) : t > 0, x \in [g(t), h(t)]\}$ to deduce  $u(t, x) \geq \underline{u}(t, x)$ in this region. 
	
	Moreover, since $q_0(\infty) = \delta$ and $q_0$ is increasing, we conclude that for large $t > 0$,
	\[
	\delta \geq \sup_{x \in \mathbb{R}} \underline{u}(t, x) \geq \inf_{x \in \mathbb{R}} \underline{u}(t, x) \geq q_0(ct - L) \to \delta\quad \text{as } t \to \infty.
	\]
	Hence
	\[
	\lim_{t \to \infty} \|\underline{u}(t,\cdot) - \delta\|_{L^\infty(\mathbb{R})} = 0.
	\]
	Since $\delta > \theta$, there exists  $t_0 > 0$ such that
	\[
{u}(t, x) \geq\underline{u}(t, x) >\theta \quad \mbox{ for }  t \geq t_0, \; x \in [g(t), h(t)].
	\]
	
	Now, let $m_0 := \min_{x \in [g(t_0), h(t_0)]} u(t_0, x)$. Then $\delta \geq m_0 >\theta$. Consider the auxiliary ODE problem:
	\[
	w' = f(w) \mbox{ for } t > t_0, \qquad w(t_0) = m_0.
	\]
	Since $f > 0$ in $(\theta, 1)$, the solution $w(t)$ is increasing in $t$, and $w(t) \to 1$ as $t \to \infty$. Thus, there exists $T_\delta \geq  t_0$ such that $w(T_\delta) = \delta$ and $m_0 \leq w(t) \leq \delta$ for $t \in [t_0, T_\delta]$. By the comparison principle over the region $\{(t, x) : t_0 \leq t \leq T_\delta, \; g(t) \leq x \leq h(t)\}$, we obtain $u(t, x) \geq w(t)$ in this region. In particular,
	\[
	u(T_\delta, x) \geq w(T_\delta) = \delta \quad \mbox{ for } x \in [g(T_\delta), h(T_\delta)].
	\]
	
	Comparing $u(t, x)$ with the constant function $\underline{u}_1(t, x) \equiv \delta$ over $\{(t, x) : T_\delta \leq t < \infty, \; g(t) \leq x \leq h(t)\}$, we conclude 
	by the usual comparison principle that $u(t, x) \geq \delta$ for $x \in [g(t), h(t)]$ and $t \in [T_\delta, \infty)$.  This proves \eqref{>delta}, and the proof of Theorem \ref{th1.3} is finished except that the existence and uniqueness of the solution to \eqref{a} will follow from a more general result proved in Section 3 below.	\hfill $\Box$

	\subsection{Numerical simulation and conjectures}
	
One naturally wonders what happens if $\delta\not\in (\theta_f^*, 1)$ in \eqref{a}. Our numerical simulation\footnote{In the simulations here, we use standard finite difference method for the reaction-diffusion equation, combined with the front tracking method for the moving boundary. Such a combination has been developed and successfully used in \cite{LL, LDL, KLSD} for related models. One could also combine finite difference with the front fixing method of \cite{PCJ}. Some comparisons of these two approaches can be found in \cite{LL}.} on \eqref{a} with monostable $f$, bistable $f$ and combustion $f$ suggests the following conjectures.

\medskip

\noindent
{\bf Conjecture 1:} If $\delta>1$ and $f$ satisfies ${\bf (f_m)}$, or ${\bf (f_b)}$, or ${\bf (f_c)}$, then there exists $\xi^*\in\mathbb R$ depending on $f$ and $u_0$ such that 
\[
\lim_{t\to\infty} g(t)=\lim_{t\to\infty} h(t)=\xi^*, \ \lim_{t\to\infty} u(t,x)=\delta \mbox{ uniformly for } x\in [g(t), h(t)].
\]

Let us note that  the above conclusions indicate that the population range vanishes as $t\to\infty$, and we will called this the {\bf vanishing case} henceforth.\footnote{This contrasts interestingly to the vanishing case for \eqref{free-bound-0}, where as $t\to\infty$, $u(t,x)\to 0$ uniformly and $[g(t), h(t)]$ converges to a finite interval $[g_\infty, h_\infty]\supset [-h_0, h_0]$.}

\medskip

In our numerical simulations, for monostable $f$, we take
\begin{equation}\label{mono}
f(u)=u(1-u);
\end{equation}
for bistable $f$, we take
\begin{equation}\label{bi}
f(u)=u(u-\theta)(1-u) \mbox{ with } \theta\in (0, 1/2);
\end{equation}
and for combustion $f$, we take,  with $ \theta\in (0,1)$ and $ \beta=0.1$,
\begin{equation}\label{comb}
f(u) =
\begin{cases} 
    0 & \text{ for } 0\leq u \leq \theta, \\ 
    \beta (u - \theta)^2 (1 - u) & \text{ for } u > \theta.
\end{cases} 
\end{equation}
For the initial function,  most of the times we take 
\begin{equation}\label{ini}
u_{0}(x)=\delta+ \alpha \cos(\frac{\pi x}{2h_{0}})\  \mbox{ with } \alpha> -\delta,\ h_0=10.
\end{equation}

Since $u_0(x)$ is even, the solution $(u(t,x), g(t),h(t))$ is also even in $x$ for any time $t>0$, and so $g(t)=-h(t)$ and $u(t, -x)=u(t,x)$. Therefore we will only consider
$u(t,x)$ for $x\geq 0$. This simplifies the simulation and discussions. Note that in such a case, we always have $\xi^*=0$ in Conjecture 1. 

A sample of our numerical simulations is given in Fig. 3, where the graph of $x\to u(t,x)$ (for several time moments) and the graph of $t\to h(t)$ are given.
In these simulations, $\theta=0.2$ was taken for bistable and combustion $f$ given in \eqref{bi} and \eqref{comb}, respectively, and $u_0(x)$ is  given by \eqref{ini} with $\alpha=0.1$.

\begin{center}
\includegraphics[width=0.8\textwidth]{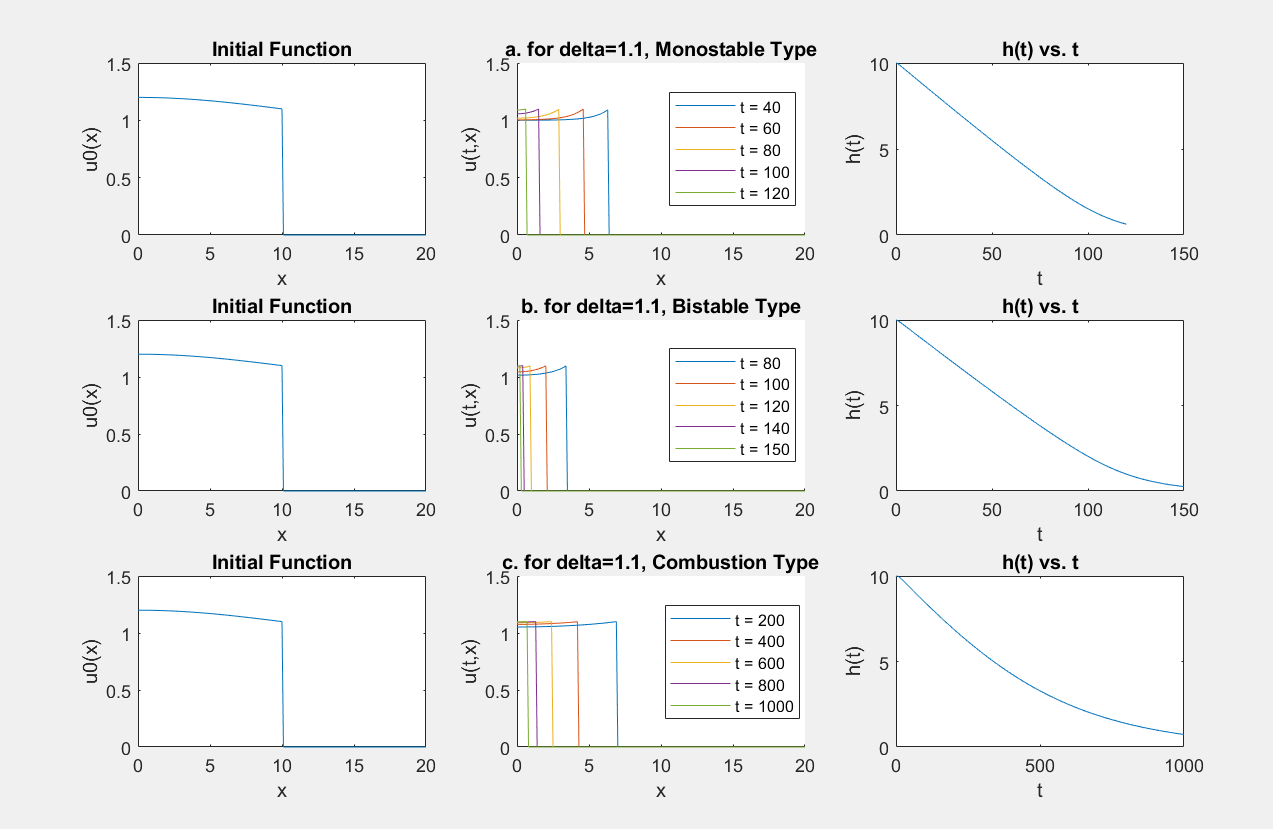}
\\
{\small Fig. 3. Simulation results with $\delta>1$: vanishing always happens.}
\end{center}

\begin{remark}\label{rm1}
Conjecture 1 will be confirmed by  Theorem \ref{th3.2} in Section 3, where we also prove that  when $\delta=1$, a new transition phenomenon between vanishing and successful spreading happens; see Theorem \ref{th3.3}.
\end{remark}

\medskip

\noindent
{\bf Conjecture 2:} When ${\bf (f_b)}$ holds and $\delta\in (\theta, \theta_f^*]$, the conclusions in Theorem \ref{th1.3} still hold except when $(u_0(x), -h_0, h_0)$ happens to be  a stationary solution of \eqref{a}.

\medskip

By a phase-plane analysis, it is easy to see that all the stationary solutions of \eqref{a} with  a bistable $f$ and $\delta\in (0, 1)$ are found as follows:

\begin{itemize}
\item[(i)] If { $\delta\in[\theta^*, 1)$}, then there are no stationary solutions.
\item[(ii)] If \ $\delta=\theta$, then the stationary solutions are: $(u, g, h)=(\theta, -h_0, h_0)$, $h_0>0$.
\item[(iii)] If { $\delta\in (0, \theta^*)\setminus \{\theta\}$}, then the solution $v(x)$ of the initial value problem
\[
dv''+f(v)=0,\; v'(0)=0, v(0)=\delta
\]
 is periodic in $x$  with $0<\min v<\theta<\max v<\theta^*$. 
 Let $L$ be  the minimal period of $v$; then
for every  integer $j\geq 1$,
\[
(v(x), -jL, jL) \mbox{ and } \Big(v(x+\frac L2), -(j-\frac 12)L, (j-\frac 12)L\Big)
\]
are stationary solutions of \eqref{a}.
There are no other stationary solutions.
\end{itemize}

Our numerical simulations for this case always exhibit successful spreading, even when we tried to start from a numerically obtained stationary solution, indicating that the stationary solutions are unstable. A sample of our simulation results for this case are given in Fig. 4, where $f(u)$ is given by \eqref{bi} and $(u_0(x), -h_0, h_0)$ is given by a numerically obtained stationary solution (for the given value of $\delta$).
\bigskip

\begin{center}
\includegraphics[width=0.8\textwidth]{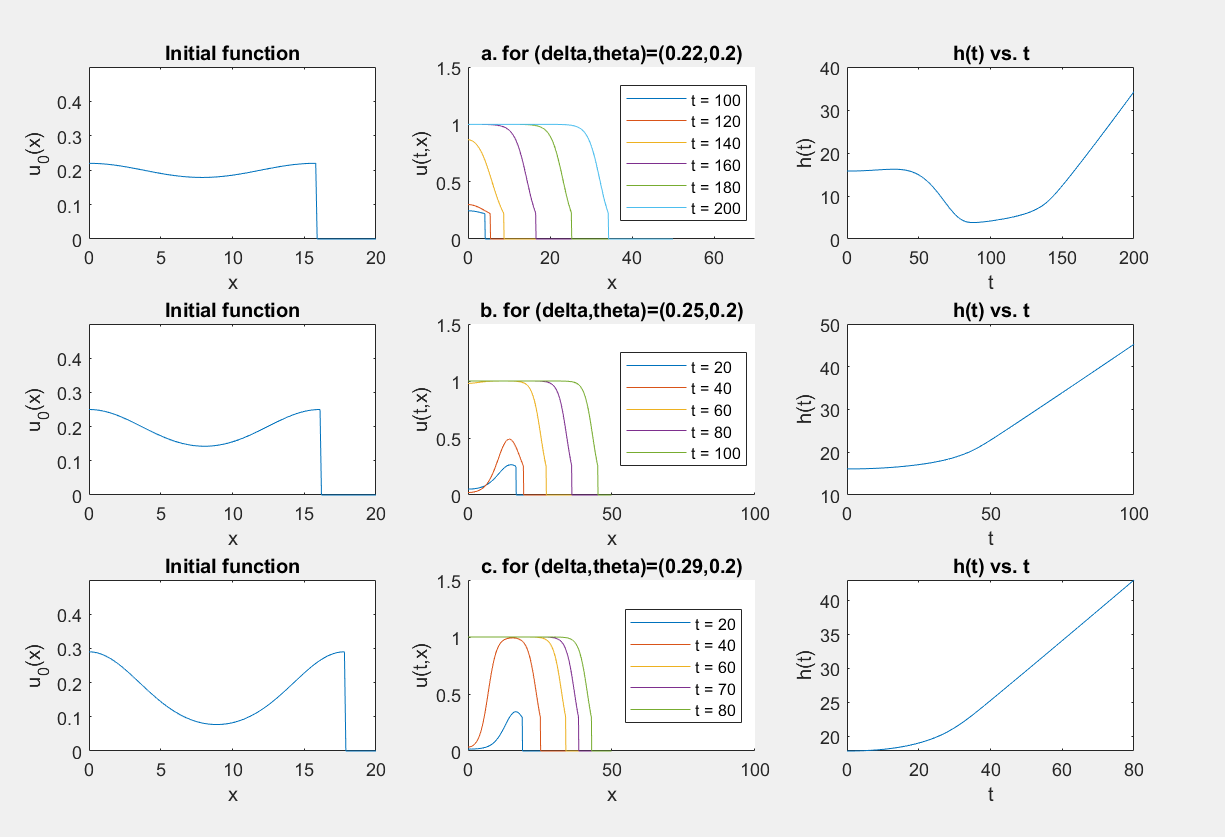}
\\
{\small Fig. 4. Simulation results for $\delta\in (\theta, \theta^*)=(0.2, 0.3101)$: spreading always happens.}
\end{center}

\bigskip

\noindent
{\bf Conjecture 3:} When {${\bf (f_b)}$ holds and  $\delta\in (0, \theta)$}, then, as $t\to\infty$, vanishing happens if the initial population is small and spreading happens if the initial population is large.

\medskip

Our simulation results for this case are listed in the table below, where $f(u)$ and $u_0(x)$ are given by \eqref{bi} and \eqref{ini}, respectively.
 \begin{table}[h]
    \centering
    {\small
    \begin{tabular}{ c c c c c}
        \toprule
      $\theta$ & $\delta$ & $\alpha$ & Longtime dynamics \\ 
        \midrule
        0.2 & 0.10 & 0.1935 & Spreading \\ 
         &  & 0.1934 & Spreading \\ 
         &  & \textbf{0.1933} & Vanishing \\ 
         &  & 0.1932 & Vanishing \\ 
        \midrule
        0.2 & 0.15 & 0.1045 & Spreading \\ 
         &  & 0.1044 & Spreading \\ 
         &  & 0.1043 & Spreading \\ 
         &  & \textbf{0.1042} & Vanishing \\ 
         &  & 0.1041 & Vanishing \\ 
        \midrule
        0.2 & 0.19 & 0.0235 & Spreading \\ 
         &  & 0.0234 & Spreading \\ 
        &  & 0.0233 & Spreading \\ 
        &  & 0.0232 & Spreading \\ 
         &  & 0.0231 & Spreading \\ 
         &  & \textbf{0.0230} & Vanishing  \\ 
        &  & 0.0229 & Vanishing \\ 
        \bottomrule
    \end{tabular}}
    \label{tab:alpha_behavior}
\end{table}

Some samples of the graphs of $x\to u(t,x)$ and $t\to h(t)$ obtained from the simulations listed in this table are given in Fig. 5.

\begin{center}
\includegraphics[width=0.8\textwidth]{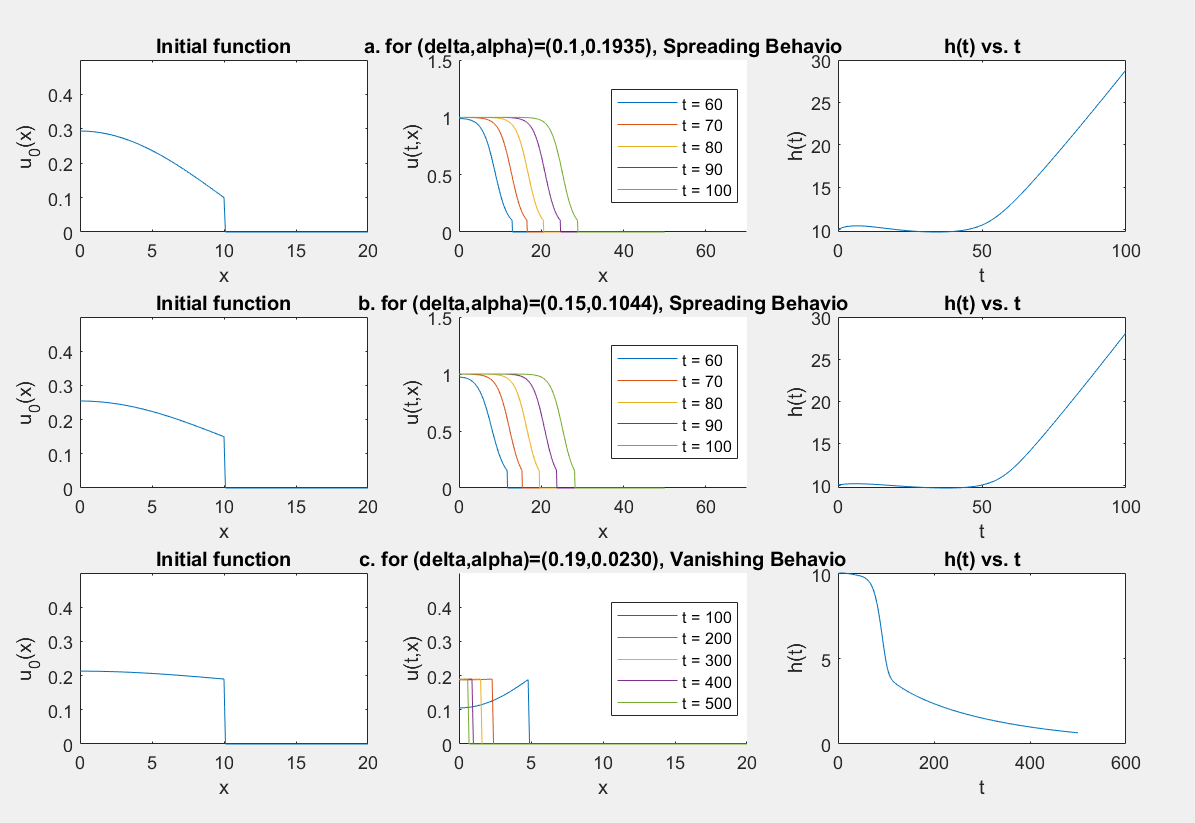}
\\
{\small Fig. 5. Simulation results for $\delta\in (0,\theta)$ with a bistable $f(u)$: either spreading or vanishing happens.}
\end{center}

\bigskip

\noindent
{\bf Conjecture 4:} When ${\bf (f_c)}$ holds and  $\delta\in (0, \theta_f^*]=(0, \theta]$, then 
\begin{itemize}
\item[(i)] $u_0\in X(h_0)$ and $u_0(x)\leq \theta$ in $[-h_0, h_0]$ imply that, as $t\to\infty$, $u(t,x)\to \delta$
and $[g(t), h(t)]\to [g_\infty, h_\infty]$ is a finite interval
 with 
 \begin{equation}\label{conserv}(h_\infty-g_\infty)\delta= \int_{-h_0}^{h_0}u_0(x)dx.
 \end{equation}
 \item[(ii)] If $u_0\in X(h_0)$ and  $u_0(x)>\theta$ for some $x\in (-h_0, h_0)$, then either  spreading happens or the above conclusions in (i) hold except that \eqref{conserv} need not be true anymore. 
 \end{itemize}
 
 \medskip

 Our simulation results for this case are listed in the table below, where $f(u)$ and $u_0(x)$ are given by, respectively, \eqref{comb} and \eqref{ini}.
 \begin{center}
 {\small
    \begin{tabular}{ c c c c c }
        \toprule
          $\theta$ & $\delta$ & $\alpha$ & Longtime dynamics \\ 
        \midrule
         0.3 & 0.28 & 5.315 & Spreading \\ 
         &  & 5.314 & Spreading \\ 
         &  & 5.313 & Spreading \\ 
         &  & 5.312 & Spreading \\ 
         &  & \textbf{5.311} & Converging to $\delta$ \\ 
         &  & 5.310 & Converging to $\delta$ \\ 
        \midrule
        0.3 & 0.29  & 4.500 & Spreading \\ 
         &  & 4.400 & Spreading \\ 
         &  & 4.300 & Spreading \\ 
         &  & \textbf{4.299} & Converging to $\delta$ \\ 
         & & 4.298 & Converging to $\delta$ \\ 
        \bottomrule
    \end{tabular}}
    \end{center}
Some samples of the graphs of $x\to u(t,x)$ and $t\to h(t)$ obtained from the simulations listed in this table are given in Fig. 6.

\begin{center}
\includegraphics[width=0.8\textwidth]{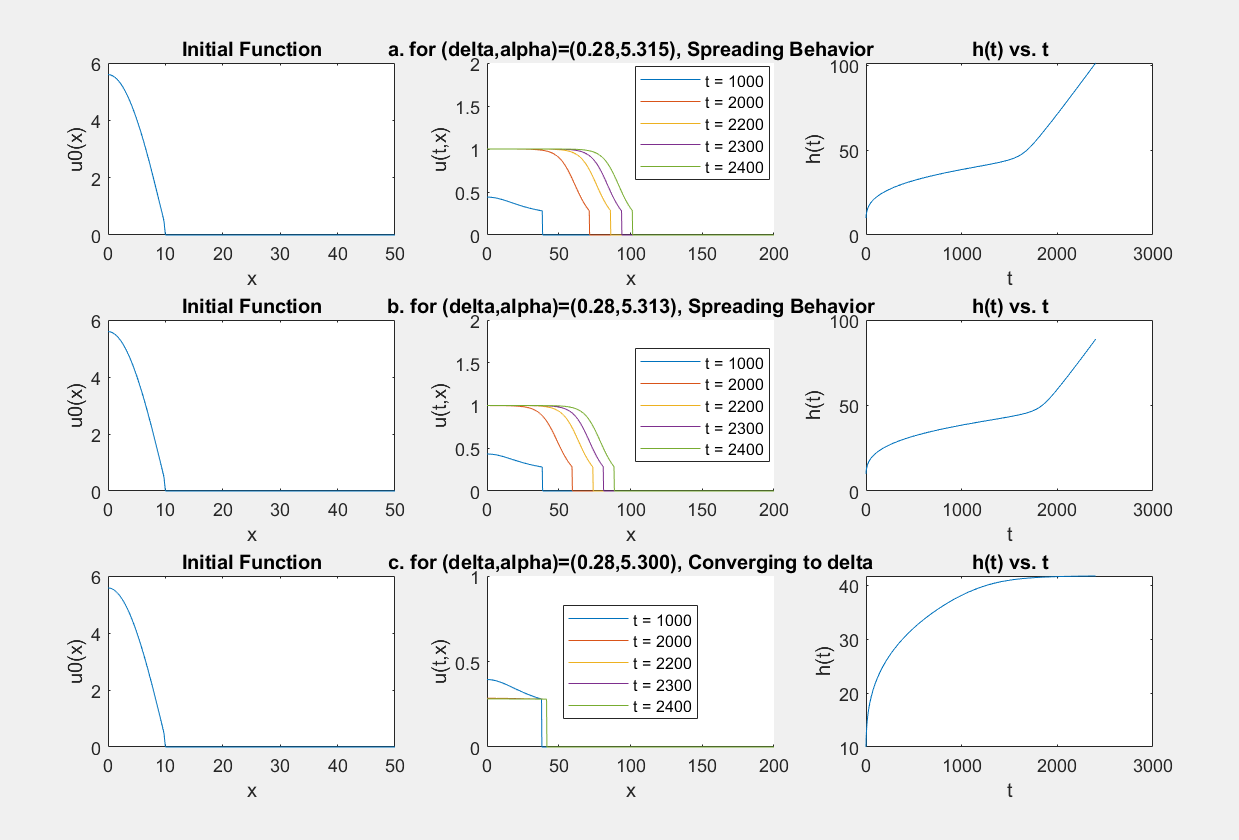}
\\
{\small Fig. 6. Simulation results for $\delta\in (0,\theta)$ with a combustion $f(u)$: either spreading or convergence to $\delta$ happens.}
\end{center}

\begin{remark}\label{rmk2} Part {\rm (i)} of Conjecture 4 will be rigorously proved in subsection 3.6.
\end{remark}

\section{Extensions to heterogeneous environment and  proof of some conjectures}
In this section, we extend the theoretical results in the previous section to situations involving heterogeneous environment, and prove several new results including in particular positive answers to Conjecture 1 and Conjecture 4 (i). 

We consider the following more general system which contains \eqref{a} as a special case. 
\begin{equation}\label{c}
\begin{cases}
u_t-du_{xx}=f(t,x, u),  & t>0,\; g(t)<x<h(t),\\
u(t,g(t))=u(t,h(t))=\delta, & t>0,\\
g'(t)=-\frac{d}{\delta}u_{x}(t,g(t)), &  t>0,\\
h'(t)=-\frac{d}{\delta}u_x(t,h(t)), &  t>0,\\
-g(0)=h(0)=h_0, u(0,x)=u_0(x), & -h_0\leq x\leq h_0.
\end{cases}
\end{equation}

For the local existence and uniqueness result, we only require $\delta>0$ and
\begin{equation}\label{f2}
	\begin{cases}
	\ f \in C^{\beta/2, \beta, 1}([0, \infty)\times \mathbb{R} \times [0, \infty)) \mbox{ for some $\beta\in (0,1)$},    f(t,x,0) \equiv  0, \mbox{ and for any $M>0$},\\
	\mbox{ $f(t,x,u)$ is Lipschitz in $u\in [0, M]$ uniformly with respect to $t\geq 0$ and $x\in\mathbb R$.}
	\end{cases}
\end{equation}

For global existence and uniqueness, we will further require
\begin{equation}\label{f3}
	f(t,x,u)<0 \mbox{ for all large $u$, say $u\geq M_0$, and all $t\geq 0$, $x\in\mathbb R$}.
\end{equation}

After proving the global existence result for \eqref{c}, we will focus on a special class of $f(t,x, u)$ with the following properties:
\begin{equation}
\label{f4}
\begin{cases} \mbox{(i) \eqref{f2} holds, $f(t,x,1)\equiv 0,\ f(t,x,u)\leq\bar f(u)< 0$ for  $t\geq 0, \ x\in\mathbb R$ and  $u>1,$}\\
\mbox{\ \ \ \ \ \  where $\bar f(u)$ is a continuous function over $[1,\infty)$, }\\
\mbox{(ii) there exists $\underline f(u)$ satisfying ${\bf (f_A)}$  such that $f(t,x,u)\geq \underline f(u)$ for $t\geq 0,\ x\in\mathbb R,\ u\geq 0$}.
\end{cases}
\end{equation}

Clearly any $f(u)$ of type ${\bf (f_A)}$ satisfies \eqref{f4}.  Let us note that any monostable  $f(u)$ also satisfies \eqref{f4} with some $\underline f(u)$ of type ${\bf (f_A)}$ with $\theta=\theta^*_{\underline f}=0$.

  \begin{theorem}\label{th3.1}
	Suppose that \eqref{f4} holds and $\delta\in (\theta_{\underline f}^*, 1)$. Then for every initial function $u_0 \in X(h_0)$, \eqref{c} has a unique solution $(u(t,x), g(t), h(t))$ defined for all $t>0$. Moreover,  as $t \to \infty$,
	\begin{align*}
		(g(t), h(t)) \to (-\infty, \infty), \quad u(t, x) \to 1 \quad \text{locally uniformly in } x \in \mathbb{R}.
	\end{align*}
Furthermore, there exists a constant $M_0>0$  such that
 \begin{align*}
 	&[g(t), h(t)]\supset [-\underline c^* t+M_0, \ \underline c^* t- M_0] \mbox{ for } t>0, \\
 			&\lim_{t \to \infty}  u(t, x)  = 1 \mbox{ uniformly for } x\in [-(\underline c^*-\epsilon)t, (\underline c^*-\epsilon)t] \mbox{ for any small } \epsilon>0,
 \end{align*}
where $\underline c^*>0$ is the unique speed of the  semi-wave problem \eqref{semi1} with $f$ replaced by $\underline f$.	
\end{theorem}

When $\delta>1$, the vanishing behaviour described in Conjecture 1 is a consequence of the following more general result.

\begin{theorem}\label{th3.2}
	Suppose that  $\delta\in (1,\infty)$, part {\rm (i)} of \eqref{f4} holds and additionally $f(t,x,u)$ is uniformly $C^{\beta/2}$-continuous in $t$ for $t\geq 0$, $x\in\mathbb R$ and $u\in [0,\delta]$.
	 Then for every initial function $u_0 \in X(h_0)$, vanishing happens to \eqref{c}, namely, there exists $\xi^*\in\mathbb R$ depending on $f$ and $u_0$ such that 
\[
\lim_{t\to\infty} g(t)=\lim_{t\to\infty} h(t)=\xi^*, \ \lim_{t\to\infty} u(t,x)=\delta \mbox{ uniformly for } x\in [g(t), h(t)].
\]
\end{theorem}

When $\delta=1$, the dynamics of \eqref{c} exhibits a transition behaviour between successful spreading (as in Theorem \ref{th3.1}) and vanishing (as in Theorem \ref{th3.2}), which is described by the following theorem.  
\begin{theorem}\label{th3.3}
	Suppose that \eqref{f4} holds and $\delta=1$. Then for every initial function $u_0 \in X(h_0)$, the unique solution $(u(t,x), g(t), h(t))$ of \eqref{c}
	satisfies, as $t \to \infty$,
	\[\begin{cases}
		[g(t), h(t)] \to [g_\infty, h_\infty] \mbox{ is a finite interval, }\smallskip\\
		\   u(t, x) \to 1 \quad \text{ uniformly for } x \in [g(t), h(t)].
	\end{cases}	
	\]
	Moreover, $h_\infty-g_\infty>0$ always holds except that, in the case $u_0(x)-1$ changes sign in $[-h_0, h_0]$, we need to assume additionally that there exists some small $\varepsilon>0$ such that
	\begin{equation}\label{additional}\varepsilon \underline f(u)\geq  \bar f(u) \mbox{ for } u\in [1, 1+\varepsilon].
	\end{equation}
\end{theorem}

The additional condition \eqref{additional} is assumed for technical reasons, which we believe is not necessary. It should be noted that both $\underline f(u)$
and $\bar f(u)$ are negative for $u>1$, and \eqref{additional} is not needed if $f=f(u)$ is independent of $(t,x)$ and is non-increasing near $u=1$.
\medskip
	
The rest of this section is organised as follows. In subsection 3.1, we prove the local existence and uniqueness for \eqref{c}, in subsection 3.2,  we present some comparison principles and obtain several crucial a priori bounds for the local solution of \eqref{c}, which will be used in the proof of the global existence and uniqueness result for \eqref{c} in subsection 3.3. Subsection 3.4 is devoted to the proof of Theorem \ref{th3.1} and subsection 3.5 gives the proof of Theorem \ref{th3.2}. In subsection 3.6, we prove Conjecture 4 part (i) and Theorem \ref{th3.3}.

\subsection{Local Existence and Uniqueness}

We prove the following rather general local existence and uniqueness theorem in this subsection.
The proof will follow the approach outlined in \cite{d2024}.

\begin{theorem}\label{th2.1}
Assuming that \eqref{f2} is satisfied, then  for any \( u_0 \in X(h_0) \) and \( \alpha \in (0, 1) \), there exists a \( T > 0 \) such that problem \eqref{c} has a unique solution
\[
(u, g, h) \in C^{(1+\alpha)/2, 1+\alpha}(\Omega_T) \times \left[C^{1+\alpha/2}([0, T])\right]^2.
\]
Moreover, 
\begin{equation*}
	\|u\|_{C^{(1+\alpha)/2, 1+\alpha}(\Omega_T)} + \|g\|_{C^{1+\alpha/2}([0, T])} + \|h\|_{C^{1+\alpha/2}([0, T])} \leq C,
\end{equation*}
and
\[
h, g \in C^{1+(1+\beta)/2}((0, T]), \quad u \in C^{1+\beta/2, 2+\beta}(\Omega_T), \quad u > 0 \text{ in } \Omega_T,
\]
where $\beta$ is from \eqref{f2},  \( \Omega_T = \{(t, x) \in \mathbb{R}^2 : t \in (0, T], x \in [g(t), h(t)]\} \), and the constants \( C \) and \( T \) depend only on \( h_0 \), \( \alpha \) and \( \|u_0\|_{C^2([-h_0, h_0])} \).
\end{theorem}

\begin{proof} The strategy of the proof is as follows. For any given continuous function pair $(g(t), h(t))$, $0\leq t\leq T$, in a suitable set, say $B^T$, and any continuous function $u(t,x)$, $0\leq t\leq T, x\in [g(t), h(t)]$,  in a certain set $Z^T$, we solve the initial boundary value problem
\begin{equation}\label{d}
\begin{cases}
v_t-dv_{xx}=f(t,x, u),  & 0<t\leq T,\; g(t)<x<h(t),\\
v(t,g(t))=v(t,h(t))=\delta, & 0<t\leq T,\\
-g(0)=h(0)=h_0, v(0,x)=u_0(x), & -h_0\leq x\leq h_0.
\end{cases}
\end{equation}
The solution $v(t,x)=v^{(u, g,h)}(t,x)$ to \eqref{d} depends on $(u,g,h)$ and induces a mapping 
\begin{equation}\label{map}
(u(t,x), g(t), h(t))\to (v(t,x), -h_0-\int_0^t \frac d\delta v_x(s, g(s))ds, h_0-\int_0^t\frac d\delta v_x(s, h(s))ds).
\end{equation}
If $(u^*, g^*, h^*)$ is a fixed point of the above mapping, then it is easy to see that $(u^*(t,x), g^*(t), h^*(t))$ solves \eqref{c} for $0\leq t\leq T$. We will show that the above mapping is indeed a contraction mapping over a suitable set of $(u,g,h)$ when $T$ is sufficiently small (and hence has a unit fixed point), which relies on the applications of parabolic $L^p$ estimates and Sobolev embeddings. 

This strategy is carried out below in four steps. In Step 1, we show that for small $T>0$, \eqref{d} can be changed to an equivalent initial boundary value problem with straight lateral boundaries, for which existing parabolic $L^p$ estimates can be readily applied. But to control the desired norm of the solution by the Sobolev embeddings  in order to apply the contraction mapping theorem later, we encounter a problem: the embedding results involve a constant which depends on $T$ and to obtain a contraction mapping we require this constant to be independent of $T$ for all small $T$. To overcome this difficulty we will use an extension trick, which is given in Step 2. Using this trick, we are able to prove the existence of a unique solution to the straight boundary problem obtained in Step 1 and show that for all small $T>0$, the mapping in \eqref{map}, but with $v$ replaced by the solution $U$ of the equivalent straight boundary problem of Step 1, is a contraction mapping, and hence has a unique fixed point; these constitute Step 3, which yields a solution to \eqref{c} for $0\leq t\leq T$ with sufficiently small $T>0$. Finally in Step 4, we explain that the obtained solution has enough regularity to be regarded as a classical solution.

{\bf Step 1}: Straightening the boundaries.

We transform the curved boundaries of \eqref{d}  into straight boundaries. For a given pair \( (g(t), h(t)) \), which are continuous in  $t$ with \( (g(0), h(0)) = (-h_0, h_0) \), we select a function \( \zeta \in C^3(\mathbb{R}) \) such that
\[
\zeta(y) = 
\begin{cases} 
	1, & \text{if } |y - h_0| < \frac{h_0}{4}, \\ 
	0, & \text{if } |y - h_0| > \frac{h_0}{2}, 
\end{cases} 
\quad \text{and } |\zeta'(y)| < \frac{5}{h_0} \text{ for all } y.
\]
Define \( \xi(y) = -\zeta(-y) \), and introduce the transformation \( (t, y) \mapsto (t, x) \) given by
\begin{equation*}
	x = \Psi(t, y) := y + \zeta(y)(h(t) - h_0) - \xi(y)(g(t) + h_0) \quad \text{for } y \in \mathbb{R}.
\end{equation*}
This transformation is a diffeomorphism for fixed \( t \in [0, T] \) from $[-h_0, h_0]$ to $[g(t), h(t)]$, provided that $T>0$ is small such that \( |h(t) - h_0| < \frac{h_0}{12} \) and \( |g(t) + h_0| < \frac{h_0}{12} \) for $t\in [0, T]$. Under the inverse of this transformation, it is clear that the curved boundaries \( x = g(t) \) and \( x = h(t) \) are mapped  to \( y = -h_0 \) and \( y = h_0 \), respectively.

Direct calculations yield the following expressions:
\begin{equation}\label{2.4}
\begin{cases}
	\displaystyle\frac{\partial y}{\partial x} = \frac{1}{1 + \zeta'(y)(h(t) - h_0) - \xi'(y)(g(t) + h_0)} =: \sqrt{A(y, h(t), g(t))}, \\[2mm]
\displaystyle 	\frac{\partial^2 y}{\partial x^2} = \frac{-\zeta''(y)(h(t) - h_0) - \xi''(y)(g(t) + h_0)}{\left[1 + \zeta'(y)(h(t) - h_0) - \xi'(y)(g(t) + h_0)\right]^3} =: B(y, h(t), g(t)), \\[2mm]
\displaystyle	\frac{\partial y}{\partial t} = \frac{-h'(t)\zeta(y) + g'(t)\xi(y)}{1 + \zeta'(y)(h(t) - h_0) - \xi'(y)(g(t) + h_0)} \\
\ \ \ \ \ =: -h'(t)C(y, h(t), g(t)) + g'(t)D(y, h(t), g(t)).
\end{cases}
\end{equation}
Define $U(t, y) := v(t, x)$ for $(t, x) \in \Omega_T$. Then  \eqref{d} is reduced to
{\small \begin{equation}\label{2.5}
\begin{cases}
	U_t - dA U_{yy} - (h'C - g'D + d B) U_y = f(t, \Psi(t, y),u(t, \Psi(t, y)), & 0 < t \leq T, \, -h_0 < y < h_0,\\
		U(t, h_0) = U(t, -h_0) = \delta, & 0 < t \leq T,\\
			h'(t) = -\frac d\delta U_y(t, h_0), \quad &0 < t \leq T,\\
				g'(t) = -\frac d\delta U_y(t, -h_0), \quad &0 < t \leq T,\\
					h(0) = h_0, \quad g(0) = -h_0, \quad
						U(0, y) = u_0(y), \quad& -h_0 \leq y \leq h_0.
\end{cases}
\end{equation}}

If $U$ solves \eqref{2.5}, then through the transformation $\Psi$ we can obtain 
\[
v(t,x)=U(t, \Psi^{-1}(t, x)), 
\]
which solves \eqref{d}. Here $\Psi^{-1}(t,\cdot)$ is the inverse of $\Psi(t,\cdot)$, i.e., 
\[\mbox{$y=\Psi^{-1}(t, \Psi(t,y))$ for $t\in [0, T],\ y\in [-h_0, h_0]$.}
\]

{\bf Step 2}: An extension trick, and application of $L^p$ theory and Sobolev embeddings.

Define $T_1$ as 
\[
T_1 := \min\left\{\frac{h_0}{12(1+h^0)}, \frac{h_0}{12(1-g^0)}\right\},
\]
where $h^0 = -\frac d\delta U_y(0, h_0)$ and $g^0 = -\frac d\delta U_y(0, -h_0)$. For $T \in (0, T_1]$, let 
\[
D_T := \{(t, y) :  0 \leq t \leq T, -h_0 \leq y \leq h_0\}.
\]
We now introduce the following function spaces:
\begin{align*}
&X_{1, T} := \{U \in C(D_T) : U(0, y) = u_0(y), \|U - u_0\|_{C(D_T)} \leq 1\},\\
&X_{2, T} := \{h \in C^1([0, T]) : h(0) = h_0, h'(0) = h^0, \|h' - h^0\|_{C([0, T])} \leq 1\},\\
&X_{3, T} := \{g \in C^1([0, T]) : g(0) = -h_0, g'(0) = g^0, \|g' - g^0\|_{C([0, T])} \leq 1\}.
\end{align*}
Clearly, the product space $X_T := \prod_{i=1}^3 X_{i, T}$ is a complete metric space when equipped with the metric: 
\[
d((U_1, h_1, g_1), (U_2, h_2, g_2)) := \|U_1 - U_2\|_{C(D_T)} + \|h_1' - h_2'\|_{C([0, T])} + \|g_1' - g_2'\|_{C([0, T])}.
\]
 For $0 < T < T_1$, we define a subspace of $X_{T_1}$, denoted by $X_{T_1}^T := \prod_{i=1}^3 X_{i, T_1}^T$, as follows
\begin{align*}
&X_{1,T_1}^{T} := \{U \in X_{1, T_1} : U(t,y) = U(T, y) \text{ for } T \leq t \leq T_1\},\\
&X_{2,T_1}^{T} := \{h \in X_{2, T_1} : h(t) = h(T) \text{ for } T \leq t \leq T_1\},\\
&X_{3,T_1}^{T} := \{g \in X_{3, T_1} : g(t) = g(T) \text{ for } T \leq t \leq T_1\}.
\end{align*}
Thus, for each $(U, h, g) \in X_T$, we can extend it to a member of $X_T^{T_1}$. Henceforth, we will identify $X_T$ with $X_T^{T_1}$.

For $(U, h, g) \in X_T = X^T_{T_1} \subseteq X_{T_1}$, using the expressions for $A$, $B$, and $C$ from \eqref{2.4}, we obtain the following estimate
\begin{equation}\label{2.6a}
\begin{aligned}
\|h'C - g'D + dB\|_{C(D_{T_1})} \leq \sup_{(y, t) \in D_{T_1}} \frac{|h'(t)\zeta(y) - g'(t)\xi(y)|}{1 + \zeta'(y)(h(t) - h_0) - \xi'(y)(g(t) + h_0)} \\
+ d \sup_{(y, t) \in D_{T_1}} \frac{|\zeta''(y)(h(t) - h_0) + \xi''(y)(g(t) + h_0)|}{\left[1 + \zeta'(y)(h(t) - h_0) - \xi'(y)(g(t) + h_0)\right]^3} \leq L_1,
\end{aligned}
\end{equation}
where $L_1$ depends on $h_0$, $h^0$, and $g^0$. Furthermore, for $(t, y) \in D_{T_1}$, it is straightforward to verify
\begin{equation}\label{2.7a}
	\frac{d}{4} \leq dA(y, h(t), g(t)) \leq 36d.
\end{equation}

For two points $Y_1 = (s_1, y_1)$ and $Y_2 = (s_2, y_2)$ in $D_{T_1}$, with parabolic distance 
\[\delta(Y_1, Y_2) = \sqrt{(y_1 - y_2)^2 + |s_1 - s_2|},
\]
 we define
 \begin{equation}\label{2.8a}
\begin{aligned}
	\omega(R) :=&\ d \max_{Y_1, Y_2 \in D_{T_1}} \sup_{\delta(Y_1, Y_2) \leq R} 
	\left| A(y_1, h(s_1), g(s_1)) - A(y_2, h(s_2), g(s_2)) \right|\\
	\leq&\ 1296 \, d \bigg| 
{\left( 1 + \zeta'(y_2)(h(s_2) - h_0) - \xi'(y_2)(g(s_2) + h_0) \right)^2}\\
	& \hspace{1.2cm} - {\left( 1 + \zeta'(y_1)(h(s_1) - h_0) - \xi'(y_1)(g(s_1) + h_0) \right)^2}
	\bigg| \\
	\leq &\ l_1 \left( \frac{h^0 + g^0 + 2}{h_0} R + h_0 \right) R \to 0 \quad \text{as } R \to 0,
\end{aligned}
\end{equation}
where $l_1$ is some suitably large constant depending on $h_0$, $h^0$, and $g^0$. 

 For $(u,g, h)$ in \eqref{d}, we write $U(t,y):=u(t,\Psi^{-1}(t,y))$ and consider the initial boundary value problem
\begin{equation}\label{2.7}
	\begin{cases}
		\bar{U}_t - dA \bar{U}_{yy} - (h'C - g'D + d B)\bar{U}_y = f(t, \Psi^{-1}(t,y),U), & -h_0 < y < h_0, \, 0 < t \leq T_1, \\
		\bar{U}(t, h_0) = \bar{U}(t, -h_0) = \delta, & 0 < t \leq T_1, \\
		\bar{U}(0, y) = u_0(y), & -h_0 \leq y \leq h_0.
	\end{cases}
\end{equation}
Using $L^p$ theory and the Sobolev embedding theorem, we conclude that \eqref{2.7} has a unique solution $\bar{U}$ satisfying
\begin{equation}\label{2.8}
	\|\bar{U}\|_{C^{(1+\alpha)/2, 1+\alpha}(D_{T_1})} \leq C_{T_1} \|\bar{U}\|_{W^{2,1}_p(D_{T_1})} \leq K_1,
\end{equation}
where $p > 3/(2 - \alpha)$, and $K_1$ depends on $p$, $\|f(\cdot, \cdot, U(\cdot,\cdot))\|_{L^p(D_{T_1})}$, $\|u_0\|_{C^2([-h_0, h_0])}$, $L_1$, $D_{T_1}$, $l_1$, and $C_{T_1}$ depends on $D_{T_1}$ and $\alpha\in (0,1)$.

Define $\bar{h}(t)$ and $\bar{g}(t)$ by
$\bar{h}'(t) = -\frac{d}{\delta} \bar{U}_y(t, h_0)$, $\bar{g}'(t) = -\frac{d}{\delta} \bar{U}_y(t, -h_0)$. 
It is clear that $\bar{h}', \bar{g}' \in C^{\alpha/2}([0, T_1])$ and satisfy
\begin{equation}
	\|\bar{h}'\|_{C^{\alpha/2}([0, T_1])}, \|\bar{g}'\|_{C^{\alpha/2}([0, T_1])} \leq K_2,
\end{equation}
where $K_2$ depends on $K_1$.

Now define the mapping $F: X_T = X_T^{T_1} \to C(D_{T_1}) \times \left[C([0, T_1])\right]^2$ by
$F(U, h, g) = (\bar{U}, \bar{h}, \bar{g})$,
and then consider its restriction on $X_T$, namely
\[
\tilde{F}(U, h, g) := F(U, h, g)\big|_{X_T}.
\]
We observe   that $(U(t,y), h(t), g(t))$ is a fixed point of the mapping $\tilde{F}$ if and only if  $(U(t, \Psi^{-1}(t, x)), h(t), g(t))$ solves \eqref{c} for $t \in [0, T]$.

{\bf Step 3:} We show that $\tilde{F}$ is a contraction mapping when $T > 0$ is sufficiently small.

For a fixed $0 < T < \min\{T_1, K_1^{-\frac{2}{1+\alpha}}, K_2^{-\frac{2}{\alpha}}\}$, we have
\begin{align*}
&\|\bar U - u_0\|_{C(D_T)} \leq T^{\frac{1+\alpha}{2}} \|\bar U\|_{C^{0, (1+\alpha)/2}(D_T)} \leq T^{\frac{1+\alpha}{2}} \|U\|_{C^{0, (1+\alpha)/2}(D_{T_1})} \leq K_1 T^{\frac{1+\alpha}{2}} \leq 1,\\
&\|\bar h'(t) - h^0\|_{C([0, T])} \leq T^{\frac{\alpha}{2}} \|\bar{h}'\|_{C^{\alpha/2}([0, T_1])} \leq T^{\frac{\alpha}{2}} \|\bar{h}'\|_{C^{\alpha/2}([0, T_1])} \leq K_2 T^{\frac{\alpha}{2}} \leq 1,\\
&\|\bar g'(t) - g^0\|_{C([0, T])} \leq T^{\frac{\alpha}{2}} \|\bar{g}'\|_{C^{\alpha/2}([0, T_1])} \leq T^{\frac{\alpha}{2}} \|\bar{g}'\|_{C^{\alpha/2}([0, T_1])} \leq K_2 T^{\frac{\alpha}{2}} \leq 1.
\end{align*}
These inequalities confirm that $\tilde{F}$ maps $X_T$ into itself.

We now proceed to show that $\tilde{F}$ is a contraction on $X_T$ for small enough $T > 0$. Let $(U_i, h_i, g_i) \in X_T = X_T^{T_1}$ for $i = 1, 2$, and define $W := \bar{U}_1 - \bar{U}_2$.
Then $W$ satisfies the following system
\begin{equation}\label{2.12}
	\begin{cases}
		W_t - dA_1 W_{yy} - (dB_1 + h_1'C_1  - g_1'D_1 ) W_y = \Phi, & -h_0 < y < h_0, \, 0 < t \leq T_1, \\
		W(t, h_0) = W(t, -h_0) = 0, & 0 < t \leq T_1, \\
		W(0, y) = 0, & -h_0 \leq y \leq h_0,
	\end{cases}
\end{equation}
where
\begin{align*}
\Phi =& (dB_1 - dB_2 + h_1'C_1 - h_2'C_2 - g_1'D_1+ g_2'D_2) \bar{U}_{2, y} \\
&+ (dA_1 - dA_2) \bar{U}_{2, yy} + f(x,t,U_1) - f(x,t,U_2).
\end{align*}
Here, $A_i := A(y, h_i, g_i)$, $B_i := B(y, h_i, g_i)$, $C_i := C(y, h_i, g_i)$, and $D_i := D(y, h_i, g_i)$ for $i = 1, 2$. A direct calculation yields
\begin{align*}
	&\|A_1 - A_2\|_{C(D_{T_1})} = \sup_{(t,y) \in D_{T_1}} \left| A(y, h_1(t), g_1(t)) - A(y, h_2(t), g_2(t)) \right| \\
	= &\sup_{(t,y) \in D_{T_1}} \bigg| 
	\frac{1}{\left(1 + \zeta'(y)(h_1(t) - h_0) - \xi'(y)(g_1(t) + h_0)\right)^2}\\
	&\quad\quad\quad\ \ - \frac{1}{\left(1 + \zeta'(y)(h_2(t) - h_0) - \xi'(y)(g_2(t) + h_0)\right)^2}
	\bigg| \\
	&\leq S_1 \left( \|h_1 - h_2\|_{C([0, T_1])} + \|g_1 - g_2\|_{C([0, T_1])} \right),
\end{align*}
where $S_1$ is a constant depending on $h_0$, $h^0$, and $g^0$, but independent of $T_1$. Similarly, we deduce
\begin{align*}
	&\|B_1 - B_2\|_{C(D_{T_1})}, \, \|C_1 - C_2\|_{C(D_{T_1})}, \, \|D_1 - D_2\|_{C(D_{T_1})}
\\ 
\leq&\  S_2
\left( \|h_1 - h_2\|_{C([0, T_1])} + \|g_1 - g_2\|_{C([0, T_1])} \right),
\end{align*}
where $S_2$ depends similarly on  parameters. Moreover, by \eqref{2.4}, 
\begin{equation*}
	\|C_2\|_{C(D_{T_1})} + \|D_2\|_{C(D_{T_1})} \leq 12 =: S_3.
\end{equation*}
By combining these estimates with  \eqref{2.8}, we deduce that for sufficiently small $T > 0$, 
\begin{equation}\label{2.14}\begin{aligned}
\|\Phi\|_{L^p(D_{T_1})} \leq &\ \|dB_1 - dB_2 + h_1'C_1 - h_2'C_2 - g_1'D_1 + g_2'D_2\|_{L^\infty(D_{T_1})} \|\bar{U}_{2,y}\|_{L^p(D_{T_1})} \\
 & + d\|A_1 - A_2\|_{L^\infty(D_{T_1})} \|\bar{U}_{2,yy}\|_{L^p(D_{T_1})} + \|f(\cdot,\cdot,U_1) - f(\cdot,\cdot,U_2)\|_{L^p(D_{T_1})}\\
\leq &\ S_1^* \left( \|U_1 - U_2\|_{C(D_{T_1})} + \|h_1 - h_2\|_{C^1([0,T_1])} + \|g_1 - g_2\|_{C^1([0,T_1])} \right),
\end{aligned}
\end{equation}
where $S_1^*$ depends on $D_{T_1}$, $K_1$, $f$, and $S_i$ for $i = 1, 2$.

Using \eqref{2.6a}, \eqref{2.7a}, \eqref{2.8a}, and \eqref{2.14}, we can apply the $L^p$ estimates to \eqref{2.12} and the Sobolev embedding theorem to derive
\begin{align*}
	&\|W\|_{C^{1+\alpha/2, 1+\alpha}(D_{T_1})} \leq C_{T_1} \|W\|_{W_p^{1,2}(D_{T_1})}\\
\leq &K_3 \left( \|U_1 - U_2\|_{C(D_{T_1})} + \|h_1 - h_2\|_{C^1([0, T_1])} + \|g_1 - g_2\|_{C^1([0, T_1])} \right),
\end{align*}
where $K_3$ is determined by $p$, $D_{T_1}$, $S_1^*$, $L_1$, $l_1$, and $C_{T_1}$. Furthermore, from the definitions of $\bar{h}_i$ and $\bar{g}_i$, we obtain
\begin{align*}
&	\|\bar{h}_1' - \bar{h}_2'\|_{C^{\alpha/2}([0, T_1])} \leq \frac{d}{\delta} \|\bar{U}_{1, y}(\cdot, h_0) - \bar{U}_{2, y}(\cdot, h_0)\|_{C^{0, \alpha/2}(D_{T_1})},\\
&\|\bar{g}_1' - \bar{g}_2'\|_{C^{\alpha/2}([0, T_1])} \leq \frac{d}{\delta} \|\bar{U}_{1, y}(\cdot, -h_0) - \bar{U}_{2, y}(\cdot, -h_0)\|_{C^{0, \alpha/2}(D_{T_1})}.
\end{align*}
Then combining the above estimates, we conclude
\begin{align*}
	&\|\bar U_1 - \bar U_2\|_{C^{(1+\alpha)/2, 1+\alpha}(D_{T_1})} + \|\bar{h}_1' - \bar{h}_2'\|_{C^{\alpha/2}([0, T_1])} + \|\bar{g}_1' - \bar{g}_2'\|_{C^{\alpha/2}([0, T_1])}\\
\leq &\ K_4 \left( \| U_1 -  U_2\|_{C(D_{T_1})} + \|h_1 - h_2\|_{C^1([0, T_1])} + \|g_1 - g_2\|_{C^1([0, T_1])} \right),
\end{align*}
where $K_4$ depends on $d$, $\delta$, and $K_3$. Finally, setting $T = \min\{ \frac{1}{2}, T_1, K_1^{-\frac{2}{1+\alpha}}, K_2^{-\frac{2}{\alpha}}, (2K_4)^{-\frac{2}{\alpha}} \}$ ensures 
\begin{align*}
&\|\bar{U}_1 - \bar{U}_2\|_{C(D_T)} + \|\bar{h}_1' - \bar{h}_2'\|_{C([0, T])} + \|\bar{g}_1' - \bar{g}_2'\|_{C([0, T])}\\
\leq&\ T^{\frac{1+\alpha}{2}} \|\bar{U}_1 - \bar{U}_2\|_{C^{\frac{1+\alpha}{2}, 1+\alpha}(D_{T_1})} + T^{\frac{\alpha}{2}} \|\bar{h}_1' - \bar{h}_2'\|_{C^{\frac{\alpha}{2}}([0, T_1])} + T^{\frac{\alpha}{2}} \|\bar{g}_1' - \bar{g}_2'\|_{C^{\frac{\alpha}{2}}([0, T_1])}\\
\leq&\ \frac{1}{2} \left( \|U_1 - U_2\|_{C(D_{T_1})} + \|h_1' - h_2'\|_{C([0, T_1])} + \|g_1' - g_2'\|_{C([0, T_1])} \right)\\
=&\ \frac{1}{2} \left( \|U_1 - U_2\|_{C(D_T)} + \|h_1' - h_2'\|_{C([0, T])} + \|g_1' - g_2'\|_{C([0, T])} \right),
\end{align*}
proving that $\tilde{F}$ is a contraction mapping on $X_T$. Hence, $\tilde{F}$ has a unique fixed point $(U, h, g)$ in $X_T$. Additionally, by the maximum principle, $U > 0$ in $[0, T] \times [-h_0, h_0]$.

{\bf Step 4:} Finally, we employ Schauder theory to derive additional regularity for the solution of \eqref{2.5} in $[-h_0, h_0] \times (0, T]$.

From Step 3, we know that $U \in C^{(1+\alpha)/2, 1+\alpha}(D_T)$ and $h, g \in C^{1+\alpha/2}([0, T])$. Consequently, the coefficients in \eqref{2.5} satisfy the following regularity properties
\begin{equation}
	dA \in C^{\alpha/2, \alpha}(D_T) \quad \text{and} \quad h'C - g'D + dB \in C^{\alpha/2, \alpha}(D_T).
	\label{2.15}
\end{equation}
However, since $u_0 \in C^2([-h_0, h_0])$, the Schauder theory cannot be applied to \eqref{2.5} direcly. To address this, we employ a standard cut-off technique. Let $\varepsilon > 0$ be any small constant, and choose $\xi^* \in C^\infty([0, T])$ such that
\[
\xi^*(t) = 
\begin{cases}
	1, & t \in [2\varepsilon, T], \\
	0, & t \in [0, \varepsilon].
\end{cases}
\]
Define $\tilde{U} := U \xi^*$. Using \eqref{2.5} we obtain 
\begin{equation}\label{2.16}
	\begin{cases}
		\tilde{U}_t = dA \tilde{U}_{yy} + (h'C - g'D + dB)\tilde{U}_y + F(t, y), & -h_0 < y < h_0, \, 0 < t < T, \\
		\tilde{U}(t, h_0) = \tilde{U}(t, -h_0) = \delta \xi^*(t), & 0 < t < T, \\
		\tilde{U}(0, y) = 0, & -h_0 \leq y \leq h_0,
	\end{cases}
\end{equation}
where $F(t, y) = U \xi_t^* + f(t,y,U) \xi^*$.

We may assume $\alpha\geq \beta$ (where $\beta$ is from \eqref{f2}), and then $F \in C^{\beta/2, \beta}([-h_0, h_0] \times [0, T])$.  In view of \eqref{2.15} and \eqref{2.7a}, we can apply the Schauder estimates  to \eqref{2.16} to obtain
\[
\|{U}\|_{C^{1+\beta/2, 2+\beta}([-h_0, h_0] \times [2\varepsilon, T])} \leq \|\tilde{U}\|_{C^{1+\beta/2, 2+\beta}([-h_0, h_0] \times [0, T])} \leq K \|F\|_{C^{\beta/2, \beta}([-h_0, h_0] \times [0, T])},
\]
where $K$ depends on $\varepsilon$, $h_0$, $\beta$, and $L_1$ in \eqref{2.6a}. Since $\varepsilon > 0$ can be arbitrarily small, it follows that $h, g \in C^{1+(1+\beta)/2}((0, T])$ and $u \in C^{1+\beta/2, 2+\beta}(\Omega_T)$.
Thus $(u, g, h)$ is a classical solution of \eqref{c} on $\Omega_T$.
\end{proof}

\subsection{Comparison principles and a priori bounds} In this subsection, we collect some results which are needed to prove the global existence and uniqueness of solutions to \eqref{c}.
We first give some comparison principles in a form that is more convenient to use than \cite[Lemma 2.2, Lemma 2.3]{d2024}.

\begin{lemma}[Comparison Principle 1]
	\label{lemma2.2b}
	Assume \eqref{f2} holds, $T \in (0, \infty)$, $g_*, \underline h,  \bar{h} \in C^1([0, T])$ with $g_*(t)<\min\{\underline h(t),\bar h(t)\}$, $\underline u \in C^{0,1}(\overline D_T) \cap C^{1, 2}( D_T)$ with $D_T = \{(t, x) \in \mathbb{R}^2 : 0 < t \leq T,\ g_*(t) < x < \underline{h}(t)\}$, $\bar u \in C^{0,1}(\overline E_T) \cap C^{1, 2}(E_T)$ with $E_T = \{(t, x) \in \mathbb{R}^2 : 0 < t \leq T,\ g_*(t) < x < \bar{h}(t)\}$, and
	\begin{equation*}\label{3.20}
	\begin{cases}
	\underline u_t - d\underline u_{xx} - f(t,x,\underline u)\leq	\bar u_t - d\bar u_{xx} - f(t,x,\bar u), & 0 < t \leq T, \, g_*(t) < x < \min\{\underline h(t),\bar{ h}(t)\}, \\
		\underline u (t, g_*(t) \leq \bar u (t, g_*(t)),\ & 0 \leq t \leq T,  \\
		\bar u = \delta, \, \bar{h}'(t) \geq -\frac d\delta \bar u_x, & 0 < t \leq T, \, x = \bar{h}(t), \\
		\underline u = \delta, \, \underline {h}'(t) \leq -\frac d\delta \underline u_x, & 0 < t \leq T, \, x = \underline{h}(t), \\
t_0\in (0, T] \mbox{ and } \underline h(t_0) < \bar{h}(t_0) \implies  \bar u (t_0, \underline h(t_0))\geq \delta, & \\
		\underline h(0) < \bar h(0), \, \underline u(0,x) \leq \bar u(0, x), & x \in [g_*(0), \underline h(0)],
	\end{cases}
	\end{equation*}
	 then
	\[
	\underline h(t) < \bar{h}(t) \ {\rm for }\ t \in (0, T], \quad \underline u(t, x) < \bar{u}(t, x)\ {\rm for }\ t \in (0, T] \ {\rm and }\ g_*(t) < x < \underline{h}(t).
	\]
\end{lemma}
\begin{proof} 
 We claim that
$\underline h(t)<\overline h(t)$ for all $t\in (0, T]$. Clearly this
is true for small $t>0$ since $\underline h(0)<\bar h(0)$. If our claim does not hold, then we can
find a first $t^*\leq T$ such that $\underline h(t)<\overline h(t)$
for $t\in (0, t^*)$ and $\underline h(t^*)=\overline h(t^*)$. It
follows that
\begin{equation}
\label{h-h}
 \underline h'(t^*)\geq \overline h'(t^*).
\end{equation}
 We now compare
$\underline u$ and $\overline u$ over the region
$$\Omega_{t^*}:=\{(t,x)\in\mathbb R^2: 0< t\leq t^*, g_* (t)< x< \underline h(t)\}.$$
 The strong maximum principle yields
$\underline u(t,x)<\overline u(t,x)$ in $\Omega_{t^*}$. Hence
$w(t,x):=\overline u(t,x)-\underline u(t,x)>0$ in $\Omega_{t^*}$ with
$w(t^*, \underline h(t^*))=0$. It then follows from the Hopf boundary lemma that $w_x(t^*, \underline h(t^*))< 0$, from which we deduce
\[
\overline h'(t^*)-\underline h'(t^*)\geq -\frac d\delta\Big[ \overline u_x(t^*, \overline h(t^*))-\underline u_x(t^*, \underline h(t^*))\Big]=-\frac d\delta w_x(t^*, \underline h(t^*))>0,
\]
which is a contradiction to 
\eqref{h-h}. This proves our claim that $\underline h(t)<\overline
h(t)$ for all $t\in (0, T]$. We may now apply the usual comparison
principle over $\Omega_{T}$ to conclude that $\underline u< \overline
u$ in $\Omega_T$.
\end{proof}

\begin{lemma}[Comparison Principle 2]	
	\label{lemma2.3b}
	Assume \eqref{f2} holds, $T \in (0, \infty)$, $\underline g, \underline h, \bar g, \bar{h} \in C^1([0, T])$ with $G(t):= \max\{\underline g(t), \bar g(t)\} <H(t):= \min\{\underline h(t),\bar{ h}(t)\}$
	for $t\in [0, T]$, $\underline u \in C^{0,1}(\bar D_T) \cap C^{1, 2}(D_T)$ with $D_T = \{(t, x) \in \mathbb{R}^2 : 0 < t \leq T,\  \underline g(t) < x < \underline{h}(t)\}$,  $\bar u \in C^{0,1}(\bar E_T) \cap C^{1, 2}(E_T)$ with $E_T = \{(t, x) \in \mathbb{R}^2 : 0 < t \leq T,\  \bar g(t) < x < \bar{h}(t)\}$, and
	\begin{equation*}\label{3.20}
	\begin{cases}
	\underline u_t - d\underline u_{xx} - f(t,x,\underline u)\leq	\bar u_t - d\bar u_{xx} - f(t,x,\bar u), & 0 < t \leq T, \, G(t) < x < H(t), \\
		\underline u = \delta, \, \underline g'(t) \geq -\frac d\delta \underline u_x, & 0 < t \leq T, \, x = \underline g(t), \\
		\underline u = \delta, \, \underline {h}'(t) \leq -\frac d\delta \underline u_x, & 0 < t \leq T, \, x = \underline{h}(t), \\
		\bar u = \delta, \, \bar g'(t) \leq -\frac d\delta \bar u_x, & 0 < t \leq T, \, x = \bar g(t), \\
		\bar u = \delta, \, \bar{h}'(t) \geq -\frac d\delta \bar u_x, & 0 < t \leq T, \, x = \bar{h}(t), \\	
t_0\in (0, T] \mbox{ and } \underline h(t_0) < \bar{h}(t_0) \implies  \bar u (t_0, \underline h(t_0))\geq \delta, &\\
t_0\in (0, T] \mbox{ and } \underline g(t_0) > \bar g(t_0) \implies  \bar u (t_0, \underline g(t_0))\geq \delta, &	\\
		[\underline g(0), \underline h(0)]\subset (\bar g(0), \bar h(0)), \, \underline u(0,x) \leq \bar u(0, x), & x \in [\underline g(0), \underline h(0)],
	\end{cases}
	\end{equation*}
	 then
	\[
	[\underline g(t), \underline h(t)]\subset (\bar g(t), \bar{h}(t)) \ {\rm for }\ t \in (0, T], \quad \underline u(t, x) < \bar{u}(t, x)\ {\rm for }\ t \in (0, T] \ {\rm and \ } \underline g(t) < x < \underline{h}(t).
	\]
\end{lemma}
\begin{proof} This is a simple variation of the proof of Lemma \ref{lemma2.2b}, and is left to the interested reader.
\end{proof}

\begin{remark} The functions $(\bar{u}, \bar{h})$ and $(\underline u, \underline h)$  in Lemmas \ref{lemma2.2b} are called a pair of upper and lower solutions to \eqref{c}; similarly, the function triples $(\bar{u}, \bar g, \bar{h})$ and $(\underline u, \underline g, \underline h)$ in Lemma \ref{lemma2.3b} are called a pair of upper and lower solutions to problem \eqref{c}. There is a symmetric version of Lemma \ref{lemma2.2b}, where the conditions on the left and right boundaries are swapped. \end{remark}
\medskip

The following result indicates that the population range $[g(t), h(t)]$ does not shrink to a single point in finite time.

\begin{lemma}\label{lemma2.4}
	\label{lemma:population_range}
	Assume \eqref{f2} and \eqref{f3} hold, $T \in (0, \infty)$, and $(u, g, h)$ solves \eqref{c} for $0 < t < T$. Then there exist positive constants $C_1(T)$ and $C_3(T)$ such that 
	\begin{align*}
		C_1(T)\leq u(t, x)\leq C_2, \ h(t)-g(t)\geq C_3(T) \quad \mbox{ for } t \in (0, T),\ x \in [g(t), h(t)],
	\end{align*}
	where  $C_2:=\max_{x\in [-h_0,h_0]}{u_0}(x)+M_0$ with $M_0$  given by \eqref{f3}. 
	\end{lemma}

\begin{proof}
	Consider the initial boundary value problem
	\begin{equation}\label{2.17}
		\begin{cases}
			v_t - dv_{xx} = f(x,t,v), & t \in (0, T), \, x \in (g(t), h(t)), \\
			v(t, g(t)) = v(t, h(t)) = \delta, & t \in (0, T), \\
			v(0, x) = v_0, & x \in [-h_0, h_0],
		\end{cases}
	\end{equation}
	with $v_0=0$ or $v_0=\max_{x\in [-h_0,h_0]}{u_0}(x)+M$. Denote by $\underline v$ the solution of \eqref{2.17} with initial condition $v_0=0$, and $\bar v$  the solution of \eqref{2.17} with initial condition $v_0=\max_{x\in [-h_0,h_0]}{u_0}(x)+M$.
	By standard comparison principle, we have
	\begin{align*}
		C_2\geq \bar v(t,x)\geq u(t, x) \geq  \underline v(t, x) >  0 \quad \mbox{ for } t \in (0, T),\ x \in [g(t), h(t)],
	\end{align*}
	with $C_2=\max_{x\in [-h_0,h_0]}{u_0}(x)+M_0$. Then by \eqref{f2}, there exists $L > 0$ depending on $C_2$ such that 
	\begin{align}\label{2.17a}
		f(x,t,u) \geq -Lu \quad \mbox{ for } t \in [0, T),\ x \in [g(t), h(t)].
	\end{align}
	Let $v_1(t)$ be the solution of the ODE problem $v_1'=-Lv_1$ with $v_1(0)=\min_{x\in [-h_0,h_0]}{u_0}(x)$. A direct calculation gives $v_1(t)=e^{-Lt}\min_{x\in [-h_0,h_0]}{u_0}(x)<\delta$ for $t> 0$. By the comparison principle, we obtain
	\begin{align*}
		u(t,x)\geq e^{-Lt}\min_{x\in [-h_0,h_0]}{u_0}(x) \quad \mbox{ for } t \in [0, T),\ x \in [g(t), h(t)].
	\end{align*}
	This proves the first part of the lemma with $C_1(T)=e^{-LT}\min_{x\in [-h_0,h_0]}{u_0}(x)$.
	
	Define
	\[
	U(t) := \int_{g(t)}^{h(t)} u(t, x) \, dx.
	\]
	Then, for $t \in (0, T)$,
	\[\begin{aligned}
	U'(t) &= h'(t)u(t, h(t)) - g'(t)u(t, g(t)) + \int_{g(t)}^{h(t)} u_t(t, x) \, dx\\
	&=\delta h'(t) - \delta g'(t) + \int_{g(t)}^{h(t)} \big[du_{xx} + f(t,x, u(t,x))\big] \, dx\\
	 &=\int_{g(t)}^{h(t)} f(t,x,u(t,x)) \, dx \geq -L \int_{g(t)}^{h(t)} u(t,x) \, dx = -L U(t).
	 \end{aligned}
	\]
	Thus, $U(t) \geq U(0)e^{-Lt} > 0 \text{ for } t \in (0, T)$.
	Recalling that $0 <u(t, x) \leq C_2 \text{ in } D_T$, we deduce
	\begin{align*}
		C_2[h(t)-g(t)]\geq U(t) \geq U(0)e^{-Lt} \geq 2h_0  \min_{x\in [-h_0,h_{0}]} u_0(x)e^{-Lt},
	\end{align*}
	which yields 
	\[
h(t)-g(t)\geq \frac{2h_0}{C_2}e^{-Lt} \min_{x\in [-h_0,h_{0}]} u_0(x).
	\]
	This completes the proof of the second part.
\end{proof}

The proof of the following a priori estimates require new techniques beyond \cite{d2024}.

\begin{lemma}\label{lemma2.5}
	Assume \eqref{f2} holds and $(u, g, h)$ is a solution to \eqref{c} defined for $t \in [0, T)$ with $T \in (0, \infty)$. Then there exists a constant $C_4(T) > 0$
	such that
	\[
	|g'(t)|, |h'(t)|\leq C_4(T) \quad {\rm for }\ t \in [0, T).
	\]
	\end{lemma}
	We will see from the proof below that $C_4(T)$ depends only on $T$, $d$, $\delta$, $u_0$ and $f$.

\begin{proof}
	In the following, we only prove the boundedness of $|h'|$, as the proof for $|g'|$ is similar.  
	
	The key is to construct a lower solution to bound  $h'(t)$ from below. 
	Define, for some constants $0 < c < 1$, $k > 0$ and $m > 0$ to be specified later, 
	\begin{align*}
		\begin{cases} 
		\theta(s) := c(e^{-s/c} - 1) \mbox{ for  } s \geq 0,\\
		\underline{u}(t,x): = \delta [\theta((h(t)-x)m) + 1]  \mbox{ for } t \geq 0, \, x \in [h(t)-k/m, h(t)].
	\end{cases}
	\end{align*}
	 We are going to show that, for suitable choices of $(c,k,m)$,  
	\begin{equation}\label{1}
		\begin{cases}
			\underline{u}_t < d \underline{u}_{xx} + f(t,x, \underline{u}), & t \in (0, T), \, x \in [h(t)-k/m, h(t)], \\
			\underline{u}(t,h(t)-k/m)  < u(t,h(t)-k/m), \ \underline{u}(t, h(t)) = \delta, & t \in [0, T), \\
			\underline{u}(0, x) \leq u_0(x), & x \in [h_0-k/m, h_0],
		\end{cases}
	\end{equation}
	and
	\begin{align}\label{2}
		h'(t) \geq -\frac{d}{\delta} \underline{u}_x(t, h(t)) = -dm, \quad t \in [0, T).
	\end{align}
	
	We first choose $c \in (0,1)$  close to $1$. Then define
	\begin{equation}\label{2.20m}
	\begin{aligned}
	&k := -c \ln\left(\frac{C_1}{c \delta} + 1 - \frac{1}{c}\right) > 0,\\
			&m := \max\left\{\frac{\max_{x \in [-h_0, h_0]} |u_0'(x)|}{\delta} e^{k/c}, \frac{2k}{C_3}, \sqrt{\frac{F_* c e^{k/c}}{d \delta (1-c)}}\right\}
	\end{aligned}
	\end{equation}
	with
	\begin{align}\label{2.23k}
		F_* := \sup_{t\geq 0, x\in\mathbb R, u \in [0, C_2]} |f(t,x,u)|,
	\end{align}
	where $C_1$, $C_2$, $C_3$ are as given in Lemma~\ref{lemma2.4}. Note that we always have $C_1 < \delta$.
	
	From Lemma~\ref{lemma2.4} and the expression for $k$, we obtain
	\begin{align*}
		&h(t) - k/m \in (g(t), h(t)) \quad {\rm for}\ t \in [0, T), \\
		&\underline{u}(t, h(t) - k/m) = C_1 \leq u(t, h(t) - k/m) \quad {\rm for}\ t \in [0, T).
	\end{align*}
	The definition of $\underline{u}$ gives $\underline{u}(t, h(t)) = \delta$. Thus, the desired relations in the second line of \eqref{1} hold. 
	
	A direct calculation gives 
	\begin{align*}
		\underline{u}_x(t,x) = m \delta e^{-[(h(t)-x)m]/c} \geq m \delta e^{-k/c} \quad \mbox{ for } t \geq 0, \, x \in [h(t)-k/m, h(t)].
	\end{align*}
	This and \eqref{2.20m} clearly yield $\underline{u}_x(0,x) \geq u_0'(x)$ for $x \in [h_0-k/m, h_0]$, which in turn implies, using $\underline{u}(0, h_0) = u_0(h_0) = \delta$, the following estimate
	\begin{align*}
		\underline{u}(0,x) \leq  u_0(x) \quad \mbox{ for } x \in [h_0-k/m, h_0),
	\end{align*}
	which is the third inequality in \eqref{1}.
	
	It remains to prove the first inequality in \eqref{1}, and in doing so we will complete the proof of the lemma.
	
	{\bf Step 1:}  We show that the first inequality in \eqref{1} holds for $t \in [0, T_1)$, where $T_1 \leq T$ is given by
	\begin{align*}
		T_1 := \sup\{t\in [0, T) : h'(s) > -\frac{d}{\delta} \underline{u}_x(s, h) = -dm \ {\rm for\ all}\ s \in [0, t)\}.
	\end{align*}
	Since $\underline{u}_x(0, h_0) = m \delta > u_0'(h_0)$, by the continuity of $h'(t)$ and $u_x(t,h(t))$ for $t\in [0, T)$, we have	
	\begin{align*}
		h'(0) = -\frac{d}{\delta} u_0'(h_0) > -\frac{d}{\delta} \underline{u}_x(0, h_0) = -dm,
	\end{align*}
	which implies that $T_1$ is well-defined, and 
	\begin{align}\label{2.20}
		h'(s) \geq -\frac{d}{\delta} \underline{u}_x(s, h) = -dm \quad {\rm for\ all}\ s \in [0, T_1).
	\end{align}
	(We will eventually show  that $T_1$ equals $T$.)
	
	A direct computation gives
	\begin{align*}
		\underline{u}_t = -h'(t) \delta m e^{-(h(t)-x)m/c}, \quad \underline{u}_x = \delta m e^{-(h(t)-x)m/c},
		 \quad \underline{u}_{xx} = \frac{\delta m^2}{c} e^{-(h(t)-x)m/c}.
	\end{align*}
	Using \eqref{2.20m}, \eqref{2.23k} and \eqref{2.20},  for $ t \in (0, T_1), \, x \in [h(t)-k/m, h(t)]$, we have
	\begin{align*}
		\underline{u}_t &\leq d \delta m^2 e^{-(h(t)-x)m/c} \leq  cd  \underline{u}_{xx}=d\underline u_{xx}-d(1-c)\frac{\delta m^2}{c} e^{-(h(t)-x)m/c}\\
		 &\leq d\underline u_{xx}-d(1-c)\frac{\delta m^2}{c} e^{-k/c} \leq d\underline u_{xx}- F_* \leq d \underline{u}_{xx} + f(t,x,\underline{u}).
	\end{align*}
	Hence, the first inequality in \eqref{1} holds for $0<t<  T_1$.
	
{\bf Step 2.} {We prove that for $t \in [0, T_1)$,
	\begin{align}\label{2.26}
		-dm \leq h'(t) \leq \frac 8 3M_*(C_2-\delta)d/\delta
	\end{align}
	for some $M_*$ defined later.}

The lower bound for $h'(t)$ has been proved in \eqref{2.20}. So we only need to prove the upper bound for $h'(t)$. Define
	\[
	\begin{aligned}
		&\Omega := \left\{(t, x) : 0 \leq t < T_1, \, h(t) - 1/(2M_*) < x < h(t)\right\}, \\
		&\bar{u}(t, x) := \frac 4 3 \left(C_2 - \delta\right)\left[2M_*(h(t)-x) - M_*^2(h(t)-x)^2\right] + \delta.
	\end{aligned}
	\]
with
\begin{align}\label{2.25M}
	M_* := \max\left\{\frac{1}{2C_3}, \sqrt{\frac{3F_*}{8d(C_2-\delta)}} + m, \frac{3\max_{x \in [-h_0, h_0]}|u_0'(x)|}{4(C_2-\delta)}\right\},
\end{align}
where $m$ and $F_*$ are given by \eqref{2.20m} and \eqref{2.23k}, respectively,   $C_2$ and $C_3$ are given by Lemma \ref{lemma2.4}. A direct calculation gives
\begin{align*}
	&\Omega \subset \{(t, x) : 0 \leq t < T_1, g(t) < x < h(t)\}, \\
	&\delta = \bar{u}(t, h(t)) \leq \bar{u}(t, x) \leq \bar{u}\left(t, h(t) - 1/(2M_*)\right) = C_2 \text{ for } (t, x) \in \Omega, \\
	&\bar{u}\left(t, h(t) - 1/(2M_*)\right) = C_2 \geq u\left(t, h(t) - 1/(2M_*)\right).
\end{align*}

Moreover, for $(t, x) \in \Omega$, by \eqref{2.20}, the estimate $\bar{u}(t,x) \leq C_2$, \eqref{2.23k} and \eqref{2.25M}, we deduce
\begin{align*}
	\bar{u}_t - d \bar{u}_{xx} - f(t,x,\bar{u}) &=\frac 4 3 \left(C_2 - \delta\right) h'(t) \left[2M_* - 2M_*^2(h(t)-x)\right] + d\frac 4 3\left(C_2 - \delta\right) 2M_*^2 - f(t,x,\bar{u}) \\
	&\geq -2dM_*m\frac 4 3(C_2 - \delta)[1 - M_*(h(t)-x)] + d\frac 4 3(C_2 - \delta) 2M_*^2 - F_* \\
	&\geq -2dM_*m\frac 4 3(C_2 - \delta) + d\frac 4 3(C_2 - \delta) 2M_*^2 - F_* \\
	& = \frac 8 3dM_*[M_*-m](C_2 - \delta) - F_* \geq 0.
\end{align*}

We show next that
\[
\bar{u}(0, x) \geq u_0(x) \text{ for } x \in \left[h_0 - 1/(2M_*), h_0\right].
\]
Indeed, by \eqref{2.25M}, for such $x$,
\begin{align*}
	\bar{u}_x(0, x) = -\frac 8 3M_*(C_2 - \delta)[1 - M_*(h_0 - x)] \leq -\frac 4 3M_*(C_2 - \delta) \leq u_0'(x).
\end{align*}
Thus, it follows from $\bar{u}(0, h_0) = u_0(h_0) = \delta$ that $\bar{u}(0, x) \geq u_0(x)$ for $x \in [h_0 - 1/(2M_*), h_0]$.

The above analysis allows us to apply the usual comparison principle over
\[
\left\{(t, x) : 0 \leq t < T_1, \, h(t) - 1/(2M_*) < x < h(t)\right\}
\]
to conclude that $u(t, x) \leq \bar{u}(t, x)$ in this region. Since $u(t, h(t)) = \bar{u}(t, h(t)) = \delta$, it follows that
\[
u_x(t, h(t)) \geq \bar{u}_x(t, h(t)) = -\frac 8 3M_*(C_2 - \delta) \quad \text{for } t \in [0, T_1).
\]
Hence,
\[
-dm \leq h'(t) = -\frac{d}{\delta} u_x(t, h(t)) \leq \frac 8 3M_*(C_2 - \delta)\frac{d}{\delta} \quad \mbox{ for } t \in (0, T_1).
\]

{\bf Step 3.} {We prove that \eqref{2}  holds, which then allows us to  repeat the argument in Step 2  with $T_1$ replaced by $T$ to conclude that \eqref{2.26} holds for all $t \in [0, T)$.}

If $T_1 = T$, then the proof is complete. Suppose now $T_1 < T$, which implies
\begin{align}\label{2.27}
	h'(T_1) = -\frac{d}{\delta} \underline{u}_x(T_1, h(T_1)) = -dm.
\end{align}
Let $w := u - \underline{u}$. Then 
\begin{equation*}
	\begin{cases}
		w_t > dw_{xx} + f(t,x,u) - f(t,x, \underline{u}) = dw_{xx} + c w, & t \in (0, T_1], \, x \in [h(t)-k/m, h(t)], \\
		w(t, h(t)-k/m) > 0, \ w(t, h(t)) = 0, & t \in [0, T_1], \\
		w(0, x) \geq 0, & x \in [h_0-k/m, h_0],
	\end{cases}
\end{equation*}
for some $c \in L^\infty$. By the comparison principle,
\[
w(t, x) > 0 \quad \mbox{ for }t \in (0, T_1], \, x \in [h(t)-k/m, h(t)).
\]
Since $w(T_1, h(T_1))=0$, we may now apply the Hopf boundary lemma to conclude that $w_x(T_1, h(T_1)) <0$. Therefore $u_x(T_1, h(T_1)) < \underline{u}_x(T_1, h(T_1))$ and
\[
h'(T_1) = -\frac{d}{\delta} u_x(T_1, h(T_1)) > -\frac{d}{\delta} \underline{u}_x(T_1, h(T_1)) = -dm,
\]
which contradicts \eqref{2.27}. Hence, $T_1 = T$ always holds. This indicates \eqref{2.26} holds for $t \in [0, T)$. The proof is complete.
\end{proof}



\subsection{Global existence}
\begin{theorem}\label{th-glob}
Assuming that \eqref{f2} and \eqref{f3} are satisfied, then the conclusions in Theorem \ref{th2.1} hold for any $T>0$.
\end{theorem}
\begin{proof}
 By Theorem \ref{th2.1}, we may assume that \eqref{c} has a unique solution $(u, g, h)$ defined on some maximal time interval $(0, T_m)$ with $T_m\in (0,\infty]$, and
\[
h, g \in C^{1+\frac{\beta}{2}}((0, T_m)), \quad u \in C^{1+\frac{\beta}{2}, 2+\beta}(\Omega_{T_m}),
\]
where
$\Omega_{T_m} := \{(t, x) : t \in (0, T_m), x \in [g(t), h(t)]\}$.
To complete the proof, we must demonstrate that $T_m = \infty$. Suppose $T_m < \infty$. Then, by the proof of Theorem \ref{th2.1}, along with Lemmas \ref{lemma2.4} and \ref{lemma2.5}, there exist positive constants $C_1, C_2=C_2(T_m)$, such that for $t \in [0, T_m)$ and $x \in [g(t), h(t)]$,
\[
0 \leq u(x, t) \leq C_1, \quad |h'(t)| + |g'(t)| \leq C_2, \quad |h(t)|, |g(t)| \leq C_2 t + h_0.
\]
For any small constant $\varepsilon > 0$, it follows from the proof of Theorem \ref{th2.1} and Lemmas \ref{lemma2.4} and \ref{lemma2.5} that $u, v \in C^{1+\beta, \frac{1+\beta}{2}}(\Omega_{T_m-\varepsilon})$. Thus, as in Step 4 of the proof of Theorem \ref{th2.1}, applying Schauder’s estimates, for any fixed $0 < T_0 < T_m - \varepsilon$, we obtain $\|u\|_{C^{2+\beta,1+\frac{\beta}{2} }(\Omega_{T_m - \varepsilon} \setminus \Omega_{T_0})} \leq Q^*$,
where $Q^*$ depends on $T_0$, $T_m$, and $C_i$ for $i = 1, 2$, but is independent of $\varepsilon$. Since $\varepsilon > 0$ can be made arbitrarily small, it follows that for any $t \in [T_0, T_m)$, 
\[
\|u(t, \cdot)\|_{C^{2+\beta}([g(t), h(t)])} \leq Q^*.
\]

By repeating the arguments used in the proof of Theorem \ref{th2.1}, we can conclude that there exists $T > 0$ small, depending on $Q^*$ and $C_i$ $(i = 1, 2)$, such that the solution to \eqref{c} with initial time $T_m - \frac{T}{2}$ can be extended uniquely to $t=T_m-\frac T2+T>T_m$, a contradiction to the definition that $T_m$ is the maximal time interval for the solution. Thus, we must have $T_m = \infty$. \end{proof}

\subsection{Proof of Theorem \ref{th3.1}} Let us first note that $f(t,x,1)\equiv 0$ and $f(t,x,u)\leq \bar f(u)<0$ for $u>1$ imply $\bar f(1)=0$.
Let $M_0=\max\{\|u_0\|_\infty, 1\}$ and $v(t)$ be the solution of 
\[
v'=\bar f(v),\ v(0)=M_0.
\]
 Since $\bar f(v)<0$ for $v>1$, we clearly have $1\leq v(t)\leq M_0$ and $v(t)\to 1$ as $t\to\infty$. Since $f(t,x,u)\leq\bar f(u)$ for $u\geq 1$, we obtain
\[
v_t-dv_{xx}=\bar f(v)\geq f(t,x,v) \mbox{ for } t>0, x\in [g(t), h(t)].
\]
We also have $v\geq 1>\delta=u$ for $x\in\{g(t), h(t)\}$ and $v(0)=M_0\geq u(0,x)$ for $x\in [-h_0, h_0]$. Therefore the standard comparison principle over the region $\{(t,x): t>0, x\in [g(t), h(t)]\}$ infers $u(t,x)\leq v(t)$ in this region. It follows that
\begin{equation}\label{limsup}
\limsup_{t\to\infty}u(t,x)\leq \lim_{t\to\infty}v(t)=1 \mbox{ uniformly for } x\in[g(t), h(t)].
\end{equation}

To bound $u(t,x)$ from below,
we first make use of \eqref{f4} to show that  there exists $T_0>0$ such that
 \begin{equation}\label{>delta1} u(t, x) \geq \delta \mbox{  for $ t\geq T_0$ and $x\in [g(t), h(t)]$}.
 \end{equation}
 Similar to the proof of Theorem \ref{th1.3}, since $\underline f$ satisfies ${\bf (f_A)}$, 
  we are able to choose a function $\hat{f} \in C^1$  sufficiently close to $\underline f$ in $L^\infty$ such that $\hat{f}(s) \leq f(s)$ for $s \geq 0$, and  $\hat f$ satisfies $\mathbf{(F_b)}$  with $(P, Q)=(\hat\theta, \delta)$ for some $\hat\theta\in [\theta, \delta)\cap (0,\delta)$.  Then, by Lemma \ref{lemma3.1}, the traveling wave problem \eqref{fg} 	has a solution pair $(c,q)=(c_0,q_0)$ with $c_0>0$ and $q_0(\cdot)$ strictly increasing.
	
	The same reasoning as in the proof of Theorem \ref{th1.3} shows that for some $L > 0$ sufficiently large,	
	\[
	\underline{u}(t, x) := \max\{q_0(ct - x - L), q_0(ct + x - L)\}
	\]
	satisfies
	 (in the weak sense)
	\[
	\underline u_t \leq  d \underline u_{xx} + \underline f(\underline u)\leq f(t,x,\underline u)  \quad \text{for } t > 0, \; x \in \mathbb{R}.
	\]
	Additionally, 	\[
		0 \leq \underline{u}(t, x) \leq  \delta \leq u(t,x) \mbox{ for } t>0,\ x\in\{g(t), h(t)\},\]
		and
		\[
		\underline{u}(0, x)  \leq u_0(x) \quad \text{for } x \in [-h_0, h_0].
	\]
	Therefore we can apply the standard comparison principle over $\{(t, x) : t > 0, x \in [g(t), h(t)]\}$ to deduce  $u(t, x) \geq \underline{u}(t, x)$ in this region. 
	
	Moreover, using
	\[
	\lim_{t \to \infty} \|\underline{u}(t,\cdot) - \delta\|_{L^\infty(\mathbb{R})} = 0,
	\]
	and $\delta > \theta$, there exists  $t_0 > 0$ such that
	\begin{equation*}\label{t0}
{u}(t, x) \geq\underline{u}(t, x) >\theta \quad \mbox{ for }  t \geq t_0, \; x \in [g(t), h(t)].
	\end{equation*}
	
	As in the proof of Theorem \ref{th1.3}, let $m_0 := \min_{x \in [g(t_0), h(t_0)]} u(t_0, x)$. Then $\delta \geq m_0 >\theta$. Consider the auxiliary ODE problem:
	\[
	w' = \underline f(w),  \quad w(t_0) = m_0.
	\]
	Since $\underline f(u)> 0$ in $(\theta, 1)$, the solution $w(t)$ is increasing in $t$, and $w(t) \to 1$ as $t \to \infty$. Thus, there exists $T_0 \geq t_0$ such that $w(T_0) = \delta$ and $m_0 \leq w(t) \leq \delta$ for $t \in [t_0, T_0]$. By the comparison principle over the region $\{(t, x) : t_0 \leq t \leq T_0, \; g(t) \leq x \leq h(t)\}$, we obtain $u(t, x) \geq w(t)$ in this region. In particular,
	\[
	u(T_0, x) \geq w(T_0) = \delta \quad \mbox{ for } x \in [g(T_0), h(T_0)].
	\]
	
	Comparing $u(t, x)$ with the constant function $\underline{u}_1(t, x) \equiv \delta$ over $\{(t, x) : T_0 \leq t < \infty, \; g(t) \leq x \leq h(t)\}$, we conclude 
	by the usual comparison principle that $u(t, x) \geq \delta$ for $x \in [g(t), h(t)]$ and $t \in [T_0, \infty)$.  This proves \eqref{>delta1},

Next using $T_0$ in \eqref{>delta1} as the initial time and consider the problem
 \begin{equation}\label{lower}
\begin{cases}
 \tilde u_t-d  \tilde u_{xx}=\underline f( \tilde u), & t>T_0,\; \tilde g(t)<x< \tilde h(t),\\
\tilde u(t, \tilde g(t))=\tilde u(t,\tilde h(t))=\delta, &  t>T_0,\\
\tilde g'(t)=-\frac{d}{\delta} \tilde u_{x}(t, \tilde g(t)),& t>T_0,\\
\tilde h'(t)=-\frac{d}{\delta}\tilde u_x(t, \tilde h(t)), & t>T_0,\\
\tilde g(T_0)=g(T_0)+\epsilon,\ \tilde h(T_0)=h(T_0)-\epsilon, \  \tilde u(T_0,x)=\delta, & g(T_0)+\epsilon\leq x\leq h(T_0)-\epsilon,
\end{cases}
\end{equation}
where $\epsilon>0$ is small.  

Since $\underline f$ satisfies ${\bf (f_A)}$ and $\delta>\theta^*_{\underline f}$, Theorem \ref{th1.3} applies to \eqref{lower}. With \eqref{>delta1} holding for all $t\geq T_0$, and $f(t,x,u)\geq \underline f(u)$, we  see that $(u,g,h)$ is an upper solution to \eqref{lower} and can 
apply Lemma \ref{lemma2.3b} to conclude that
\[
[\tilde g(t), \tilde h(t)]\subset (g(t), h(t)),\ \ \tilde u(t,x)\leq u(t,x) \mbox{ for } t\geq T_0, \ x\in [\tilde g(t), \tilde h(t)].
\]
The conclusions of the theorem now follow directly from this, \eqref{limsup} and the results in Theorem \ref{th1.3} for $(\tilde u, \tilde g, \tilde h)$.
\hfill $\Box$

\subsection{Proof of Theorem \ref{th3.2}} We complete the proof by three steps.\medskip

\noindent
{\bf Step 1.} We show that there exists $T_0>0$ such that $u(t,x)\leq \delta$ for $t\geq T_0$ and $x\in [g(t), h(t)]$.

Let $m_0=\|u_0\|_\infty+1$ and consider the ODE problem
\[
w'=\bar f(w),\ w(0)=m_0.
\]
Since $\bar f(1)=0$, $\bar f(u)<0$ for $u>1$ and $m_0\geq 1+\delta$, the unique solution $w(t)$ is strictly decreasing and $w(t)\to 1$ as $t\to\infty$. Therefore there exists a unique $T_0>0$ such that $w(T_0)=\delta\in (1, m_0)$. We may now compare $u(t,x)$ with $w(t)$ in the region $\{(t,x): t\in [0, T_0], \ x\in [g(t), h(t)]\}$ by the usual comparison principle, and conclude that $u(t,x)\leq w(t)$ in this region. It follows that $u(T_0, x)\leq \delta$ for $x\in [g(T_0, h(T_0)]$. Since $f(t,x,\delta)\leq \bar f(\delta)<0$, we see that over the region $\{(t,x): t\geq T_0, \ x\in [g(t), h(t)]\}$, $\bar u\equiv\delta$ is always an upper solution to the initial boundary value problem satisfied by $u(t,x)$, and therefore $u(t,x)\leq \delta$ in this region.
\medskip

\noindent
{\bf Step 2.} We show that there exists $\xi^*\in\mathbb R$  such that 
\begin{equation}\label{xi*}
\lim_{t\to\infty} g(t)=\lim_{t\to\infty} h(t)=\xi^*.
\end{equation}

The conclusion proved in Step 1 above implies that for $t\geq T_0$,
\[
h'(t)=-\frac d\delta u_x(t, h(t))\leq 0,\ g'(t)=-\frac d\delta u_x(t, g(t))\geq 0.
\]
Therefore
\[
h_\infty:=\lim_{t\to\infty} h(t) \mbox{ and } g_\infty:=\lim_{t\to\infty} g(t) \mbox{ exist, and } [g_\infty, h_\infty]\subset [g(T_0, h(T_0)].
\]
It remains to show $g_\infty=h_\infty$. 

Arguing indirectly we assume  $h_\infty>g_\infty$ and seek  a contradiction. Define
\[   y=\frac {[h(t)-x] g_\infty+[x-g(t)]h_\infty}{h(t)-g(t)},\ \ U(t,y)=u(t,x).
\]
By direct calculation we see that  $U$ satisfies
\begin{equation}\label{U}
\begin{cases}
U_t-d\Big[\frac{h_\infty-g_\infty}{h(t)-g(t)}\Big]^2U_{yy}+\frac{h'(t)g_\infty-g'(t)h_\infty-[h'(t)-g'(t)]y}{h(t)-g(t)}U_y=F(t,y,U), & t>0,\ y\in [g_\infty, h_\infty],\\
U(t, g_\infty)=U(t, h_\infty)=\delta, & t>0,\\
h'(t)=-\frac d\delta \frac{h_\infty-g_\infty}{h(t)-g(t)}U_y(t, h_\infty), & t>0,\\
g'(t)=-\frac d\delta \frac{h_\infty-g_\infty}{h(t)-g(t)}U_y(t, g_\infty), & t>0,\\
U(0, y)=u_0\Big(\frac{2h_0y-(g_\infty+h_\infty)h_0}{h_\infty-g_\infty}\Big), & y\in [g_\infty, h_\infty],
\end{cases}
\end{equation}
where
\[
F(t,y,U):=f(t, \frac{h(t)(y-g_\infty)+g(t)(h_\infty-y)}{h_\infty-g_\infty}, U).
\]
We easily see that $\|F(\cdot, \cdot, U(\cdot,\cdot))\|_{L^\infty([0,\infty)\times [g_\infty, h_\infty])}<\infty$.  Since $h'(t)\geq 0\geq g'(t)$ for $t\geq T_0$, from Step 2 in the proof of Lemma \ref{lemma2.5} we see that
$\|h'\|_{L^{\infty}([0,\infty))}+\|g'\|_{L^\infty([0,\infty))}<\infty$. Therefore, in the first line of \eqref{U}, the coefficient of $U_{yy}$  is uniformly continuous, and the coefficient of $U_y$ belongs to $L^\infty([0,\infty)\times [g_\infty, h_\infty])$.
 We may now apply standard $L^p$ theory to the initial boundary value problem consisting of the first, second and fifth lines of \eqref{U} to conclude that, for any $p>1$, 
\[
\|U\|_{W^{1,2}_p([n, n+2]\times [g_\infty, h_\infty])}\leq C_p
\]
for all integers $n\geq 1$ and some $C_p>0$ independent of $n$. Taking $p$ sufficiently large and use the Sobolev embedding theorem, we see from
the third and fourth equations in \eqref{U} that $h'(t)$ and $g'(t)$ are uniformly continuous in $t$ for $t\geq 1$. This implies that $h'(t)\to 0$ and $g'(t)\to 0$ as $t\to\infty$,
since otherwise,  say there exists $s_n$ increasing to $\infty$ as $n\to\infty$  such that $h'(s_n)\geq \sigma>0$ for all $n$, then we can find $\epsilon>0$ small so that $h'(s)\geq \sigma/2$ for $s\in [s_n-\epsilon, s_n+\epsilon]$ and all $n\geq 1$, which together with $h'(t)\geq 0$ for $t>T_0$ implies $h_\infty=\infty$.

Let $\{t_n\}$ be an arbitrary sequence increasing to $\infty$ as $n\to\infty$, and define
\[
U_n(t,y):=U(t_n+t, y).
\]
Then applying the $L^p$ estimates to the initial boundary value problem satisfied by $U_n$ (which is a simple variation of the first, second and fifth lines of \eqref{U}) and using  the Sobolev embedding theorem,
we see that, for any $\alpha\in [\beta,1)$, subject to a subsequence, 
$U_n\to \tilde U$ in $C_{loc}^{(1+\alpha)/2, 1+\alpha}(\mathbb{R}\times [g_\infty, h_\infty])$. By Step 1 we may assume $U_n\leq \delta$ for all $n\geq 1$ and hence $\tilde U\leq \delta$.

We claim that
 $\tilde U$ is a classical solution of
\begin{equation}\label{tilde-U}
\begin{cases}
\tilde U_t-d\tilde U_{yy}=\tilde F(t, y, \tilde U), & t\in\mathbb R,\ y\in [g_\infty, h_\infty],\\
\tilde U(t, g_\infty)=\tilde U(t, h_\infty)=\delta, & t\in\mathbb R,\\
\tilde U_y(t, g_\infty)=\tilde U_y(t, h_\infty)=0, & t\in\mathbb R,
\end{cases}
\end{equation}
where $\tilde F(t,y,\xi)$ is a function in $C^{\beta/2, \beta, 1-0}(\mathbb R\times [g_\infty, h_\infty]\times [0, \delta])$, with $\xi\to \tilde F(t,y,\xi)$   Lipschitz continuous over $[0, \delta]$ uniformly for $t\in\mathbb R$ and $y\in [g_\infty, h_\infty]$, and $\tilde F(t,y, \xi)\leq \bar f(\xi)<0$ for $\xi\in (1, \delta]$ and all $t\in\mathbb R,\ y\in [g_\infty, h_\infty]$.

Indeed, by our assumptions on $f(t,x,u)$ we know that, for any $T>0$, there exists a positive integer $n_T$ such that
\[
\hat F_n(t, y, \xi):=F(t_n+t, y,\xi), \ n\geq n_T
\]
form a sequence in $C([-T, T]\times [g_\infty, h_\infty]\times [0, \delta])$ which is uniformly bounded and equicontinuous. Therefore by passing to a subsequence,
$\hat F_n(t,y,\xi)\to \tilde F_T(t,y,\xi)$  as $n\to\infty$ in the space $C([-T, T]\times [g_\infty, h_\infty]\times [0, \delta])$. Taking a sequence $T_k\to\infty$ as $k\to\infty$ and using a standard diagonal process of selecting subsequences, we can find a subsequence of $\{t_n\}$, still denoted by $\{t_n\}$ for convenience of notation, and a function $\tilde F\in C(\mathbb R\times [g_\infty, h_\infty]\times [0, \delta])$
such that, as $n\to\infty$,
\[\hat  F_n(t, y, \xi)\to \tilde F(t,y, \xi) \mbox{ uniformly in $[-T, T]\times [g_\infty, h_\infty]\times [0, \delta]$ for every $T>0$.}
\]
Since  
\[\mbox{ $U_n(t,y)\to \tilde U(t, y)$ and }\ 
Y_n(t,y):=\frac{h(t_n+t)(y-g_\infty)+g(t_n+t)(h_\infty-y)}{h_\infty-g_\infty}\to y
\]  
uniformly in $[-T, T]\times [g_\infty, h_\infty]$ as $n\to\infty$,  we see that for any $T>0$,
\[
\begin{aligned}
0&\leq |\hat F_n(t, Y_n(t,y), U_n(t,y))-\tilde F(t,y, \tilde U(t,y))|\\
&\leq  |\hat F_n(t, Y_n(t,y), U_n(t,y))-\tilde F(t,Y_n(t,y),  U_n(t,y))|\\
& \ \ \ \ \ \ + |\tilde F(t, Y_n(t,y), U_n(t,y))-\tilde F(t,y, \tilde U(t,y))|\\
& \to 0 \mbox{ as $n\to\infty$ uniformly in $[-T, T]\times [g_\infty, h_\infty]\times [0, \delta]$.}
\end{aligned}
\]
One can now easily see from the above fact for $\hat F_n(t, Y_n(t,y), U_n(t,y))$ and the equations satisfied by $U_n(t,y)$  that for any $p>1$,
 $\tilde U$ is a $W^{1,2}_{p, loc}(\mathbb R\times [g_\infty, h_\infty])$ solution of
\eqref{tilde-U}. 

Since $\xi\to f(t,x,\xi)$ is Lipschitz for $\xi\in [0, \delta]$ uniformly for $t\geq 0$ and $ x\in\mathbb R$, there exists a constant $L>0$ such that 
\[
|\hat F_n(t,y,\xi_1)-\hat F_n(t,y,\xi_2)|\leq L|\xi_1-\xi_2| \mbox{ for } t\geq -t_n, y\in [g_\infty, h_\infty],\ \xi_1,\xi_2\in [0, \delta].
\]
Letting $n\to\infty$ we deduce
\[
|\tilde F(t,y,\xi_1)-\tilde F(t,y,\xi_2)|\leq L|\xi_1-\xi_2| \mbox{ for } t\in\mathbb R, y\in [g_\infty, h_\infty],\ \xi_1,\xi_2\in [0, \delta].
\]
Similarly, we can show $\tilde F(t,y,\xi)$ is $C^{\beta/2, \beta}$ in $(t,y)$ for $ t\in\mathbb R, y\in [g_\infty, h_\infty],\xi\in[0,\delta]$. Moreover, from $f(t,x,u)\leq \bar f(u)$ we deduce 
\[
\hat F_n(t,y,\xi)\leq \bar f(\xi)  \mbox{ for } t\geq -t_n, y\in [g_\infty, h_\infty],\ \xi \in [1, \delta].
\]
Letting $n\to\infty$ we obtain $\tilde F(t,y,\xi)\leq\bar f(\xi)$ for $t\in\mathbb R,\ y\in [g_\infty, h_\infty],\ \xi \in [1, \delta]$.
Thus $\tilde F$ has all the claimed properties.

With these properties of $\tilde F$, we may now use standard regularity iteration to see that $\tilde U$ is a classical solution of \eqref{tilde-U}. Recall that  Step 1 implies  $\tilde U(t,y)\leq \delta$, and $\tilde F(t,y,\delta)\leq \bar f(\delta)<0$ further implies that  we can use  the usual strong maximum principle to conclude that
$\tilde U (t, y)<\delta$ for $t\in\mathbb R$ and $y\in (g_\infty, h_\infty)$. Then the Hopf boundary lemma applied to the function $W(t,y):=\delta-\tilde U(t,y)$ at the boundary points $(t, g_\infty)$ and $(t, h_\infty)$ leads to $\tilde U_y(t, g_\infty)<0<\tilde U_y(t, h_\infty)$, which contradicts the third equation in \eqref{tilde-U}. 
This contradiction proves the desired fact that $g_\infty=h_\infty$, and Step 2 is now finished.
\medskip

\noindent
{\bf Step 3.} We prove that $u(t,x)\to\delta$ as $t\to\infty$ uniformly for $x\in [g(t), h(t)]$.

Since by Step 1, $u(t,x)\leq \delta$ for $t\geq T_0$ and $x\in [g(t), h(t)]$, we only need to bound $u(t,x)$ from below. We do this by using suitable lower solutions.
By a translation of $x$ we may assume that $\lim_{t\to\infty}g(t)=\lim_{t\to\infty} h(t)=0$. By the assumptions on $f(t,x,u)$ we can find $M>0$ such that 
\[
\mbox{$f(t,x,\xi)\geq -M\xi$ for all $t\geq 0, x\in\mathbb R$ and $\xi\in [0, \delta]$.}
\]

 For any small $\epsilon>0$, we consider the auxiliary problem
\begin{equation}\label{-M}
\begin{cases} v_t-dv_{xx}=-Mv, & t>0, x\in [-\epsilon,\epsilon],\\
v(t,\pm \epsilon)=\delta, & t>0,\\
v(0, x)=0, & x\in [-\epsilon,\epsilon].
\end{cases}
\end{equation}
Since $0$ is a lower solution of the stationary problem of \eqref{-M}, the unique solution $v(t,x)$ is increasing in $t$. As $\delta$ is an upper solution of \eqref{-M}, we further have $v(t,x)\leq \delta$ for all $t>0$ and $x\in [-\epsilon,\epsilon]$. Therefore $v_*(x):=\lim_{t\to\infty} v(t,x)$ exists, and by standard regularity considerations $v_*$ is a stationary solution of \eqref{-M}, and $0<v_*(x)\leq\delta$. The maximum principle indicates that $v_*$ is the unique stationary solution of \eqref{-M}. Define
\[
v_\epsilon(x):=v_*(\epsilon x) \mbox{ for } x\in [-1,1].
\]
Clearly
\[
-d v_\epsilon''=-\epsilon^2 M v_\epsilon \mbox{ in } (-1, 1),\ v_\epsilon(\pm 1)=\delta.
\]
Since $\epsilon^2Mv_\epsilon\to 0$ in $L^\infty([-1,1])$ as $\epsilon \to 0$, it follows that $v_\epsilon\to v_0$ in $W^{2,p}([-1,1])$ for any $p>1$, and $v_0$ solves
\[
-dv_0''=0,\ v_0(\pm 1)=\delta.
\]
Hence $v_0\equiv \delta$ and $v_\epsilon(x)\to\delta$ as $\epsilon\to 0$ uniformly for $x\in [-1,1]$.

Let $T_\epsilon>T_0$ satisfy
\[
[g(t), h(t)]\subset (-\epsilon, \epsilon) \mbox{ for } t\geq T_\epsilon.
\]
Then we can compare $u(t,x)$ with $v(t-T_\epsilon, x)$ over the region $\{(t,x): t\geq T_\epsilon, x\in [g(t), h(t)]\}$ by the usual comparison principle to conclude that
$
u(t,x)\geq v(t-T_\epsilon,x)$ in this region. It follows that 
\[
\liminf_{t\to\infty} u(t,x)\geq \lim_{t\to\infty} v(t-T_\epsilon,x)=v_*(x)=v_\epsilon(x/\epsilon) \mbox{ uniformly for } x\in [g(t), h(t)].
\]
Letting $\epsilon\to 0$ and using $v_0\equiv\delta$ we deduce
\[
\liminf_{t\to\infty} u(t,x)\geq \delta \mbox{ uniformly for } x\in [g(t), h(t)].
\]
This completes Step 3, and the proof of Theorem \ref{th3.2} is finished. \hfill $\Box$

\subsection{Proof of Conjecture 4 part (i) and Theorem \ref{th3.3}}  The proof of these two results share several ideas and techniques.
\subsubsection{Proof of Conjecture 4 part (i)} The case $\delta=\theta$ is simpler and will be proved after the case $\delta\in (0, \theta)$.
The proof of this latter case consists of the following three steps.
\medskip

\noindent
{\bf Step 1.} We prove that there exists $M>0$ such that $[g(t), h(t)]\subset (-M, M)$ for all $t>0$ and $\limsup _{t\to\infty} u(t,x)\leq \delta$ uniformly for $x\in[g(t), h(t)]$.

Clearly $\tilde f(u):=f(u+\delta)$ is  a combustion type function with ignition temperature $\tilde\theta=\theta-\delta$ and $\tilde f(1-\delta)=0$, $\tilde f(u)<0$ for $u>1-\delta$. Since $u_0(x)\leq\theta$ for $x\in [-h_0, h_0]$ we can find $\bar u_0(x)\in X(h_0+1)$ such that 
\[
\theta\geq \bar u_0(x)\geq u_0(x) \mbox{ for } x\in [-h_0, h_0],\ \bar u_0(x)>\delta \mbox{ for } x\in (-h_0, h_0).
\]

We now consider the auxiliary free boundary problem
 \begin{equation}\label{tilde-0}
  \begin{cases}
  \tilde u_t-d\tilde u_{xx}=\tilde f(u), & t>0, \ \tilde g(t)<x<\tilde h(t),\\
   \tilde u(t,\tilde g(t))= \tilde u(t,\tilde h(t))=0,&t>0,\\
    \tilde g'(t)=-\frac d\delta \tilde u_x(t, \tilde g(t)), & t>0,\\
    \tilde  h'(t)=-\frac d\delta \tilde u_x(t, \tilde h(t)),&t>0,\\
   \tilde  u(0,x)=\bar u_0(x)-\delta,\ -\tilde g(0)=\tilde h(0)=h_0+1,& \tilde g(0)\leq x\leq\tilde h(0).
  \end{cases}
  \end{equation}
By \cite{dl2015}, there are three possibilities for the unique solution $(\tilde u, \tilde g, \tilde h)$ of \eqref{tilde-0} as $t\to\infty$: (i) spreading, (ii) vanishing and (iii) transition. Since $\bar u_0(x)-\delta\leq \tilde\theta=\theta-\delta$, it follows from the comparison principle that $\tilde u(t,x)<\tilde\theta$ for $t>0$ and $x\in [\tilde g(t), \tilde h(t)]$. Therefore $\tilde f(\tilde u(t,x))\equiv 0$, which implies, by the comparison principle again, $\tilde u(t,x)\leq \max_{x\in [\tilde g(t_0), \tilde h(t_0)]}\tilde u(t_0,x)<\tilde\theta$ for $t>t_0>0$. Hence only vanishing is possible for $(\tilde u, \tilde g,\tilde h)$. It follows that 
\[
\lim_{t\to\infty}[\tilde g(t), \tilde h(t)]=[\tilde g_\infty, \tilde h_\infty] \mbox{ is a finite interval, } [\tilde g(t), \tilde h(t)]\subset (\tilde g_\infty,\tilde h_\infty) \mbox{ for } t>0,
\]
and 
\[
\tilde u(t,x)\to 0 \mbox{ as } t\to\infty \mbox{ uniformly for } x\in [\tilde g(t),\tilde h(t)].
\]

On the other hand, the standard comparison principle implies $u(t,x)\leq \theta$ for $t\geq 0$ and $x\in [g(t), h(t)]$. So $f(u(t,x))\equiv 0$ and it is easy to check that $(u, g, h)$ and $(\tilde u+\delta, \tilde g, \tilde h)$ form a pair of upper and lower solutions as described in Lemma \ref{lemma2.3b}, which allows us to conclude that
\[
[g(t), h(t)]\subset (\tilde g(t), \tilde h(t))\subset (\tilde g_\infty, \tilde h_\infty),\ u(t,x)\leq \tilde u(t,x)+\delta \mbox{ for } t>0,\ x\in [g(t), h(t)].
\]
The desired conclusions clearly follow directly from the above estimates.
\medskip

\noindent
{\bf Step 2.} We show that $u(t,x)\to\delta$ as $t\to\infty$ uniformly for $x\in [g(t), h(t)]$.

Let $M>0$ satisfy $[g(t), h(t)]\subset (-M, M)$ for $t>0$ as stated in Step 1. 
Then consider the auxiliary initial boundary value problem
\[
\begin{cases}
v_t-d v_{xx}=0,& t>0,\ -M<x<M,\\
v(t, \pm M)=\delta,& t>0,\\
v(0, x)=0, & x\in [-M, M].
\end{cases}
\]
One easily sees that $v(t,x)$ is increasing in $t$ and $v(t,x)\to\delta$ as $t\to\infty$ uniformly for $x\in [-M, M]$. 
Moreover, we can compare $u(t,x)$ with $v(t,x)$ over the region $\{(t,x): t>0, x\in [g(t), h(t)]\}$ and use the usual comparison principle to conclude that
$u(t,x)\geq v(t,x)$ in this region. It follows that 
\[
\liminf_{t\to\infty}u(t,x)\geq \lim_{t\to\infty} v(t,x)=\delta \mbox{ uniformly for } x\in [g(t), h(t)].
\]
This and the second conclusion in Step 1 imply that $u(t,x)\to\delta$ as $t\to\infty$ uniformly for $x\in [g(t), h(t)]$. 
\medskip

\noindent
{\bf Step 3.} We prove that
 $g_\infty:=\lim_{t\to\infty}g(t)$ and $h_\infty:=\lim_{t\to\infty}h(t)$ exist, and 
 \[
 \int_{-h_0}^{h_0}u_0(x)dx=\delta (h_\infty-g_\infty).
 \]

 By Step 2, for any given small $\epsilon>0$, there exists $T=T_\epsilon>0$ such that
 \[
 u(t,x)\in [\delta-\epsilon, \delta+\epsilon]\subset (0, \theta) \mbox{ for } t\geq T,\ x\in [g(t), h(t)].
 \]
  Now we consider the following auxiliary problem
\begin{equation}\label{hat-u}
 \begin{cases}
  \hat u_t-d\hat u_{xx}=0, & t>0, \ \hat g(t)<x<\hat h(t),\\
   \hat u(t,\hat g(t))= \hat u(t,\hat h(t))=\delta,&t>0,\\
    \hat g'(t)=-\frac d\delta \hat u_x(t, \hat g(t)), & t>0,\\
    \hat  h'(t)=-\frac d\delta \hat u_x(t, \hat h(t)),&t>0,\\
   \hat  u(0,x)=\hat u_0(x),\ \hat g(0)=g(T)-\epsilon, \hat h(0)=h(T)+\epsilon,& \hat g(0)\leq x\leq \hat h(0),
  \end{cases}
\end{equation}
where $\hat u_0$ is a $C^2$ function satisfying
\[
\hat u_0(\hat g(0))=\hat u_0(\hat h(0))=\delta,\ \max\{\delta, u(T, x)\} \leq \hat u_0(x)\leq \delta+\epsilon \mbox{ for } x\in [\hat g(0), \hat h(0)].
\]
We easily see that the unique solution $(\hat u, \hat g, \hat h)$ of \eqref{hat-u} satisfies $\delta+\epsilon\geq \hat u(t,x)\geq \delta$ for all $t>0$ and $x\in [\hat g(t), \hat h(t)]$.
It follows that $\hat h'(t)\geq 0\geq \hat g'(t)$ for $t>0$.
Define
\[
\hat U(t):=\int_{\hat g(t)}^{\hat h(t)} \hat u(t,x)dx.
\]
Then, for $t >0$,
	\[\begin{aligned}
	\hat U'(t) &= \hat h'(t)\hat u(t, \hat h(t)) - \hat g'(t)\hat u(t, \hat g(t)) + \int_{\hat g(t)}^{\hat h(t)} \hat u_t(t, x) \, dx\\
	&=\delta \hat h'(t) - \delta \hat g'(t) + \int_{\hat g(t)}^{\hat h(t)} d \hat u_{xx}  \, dx\\
	 &=\delta \hat h'(t) - \delta\hat g'(t) +d\hat u_x(t, \hat h(t))-d\hat u_x(t, \hat g(t))=0.
	 \end{aligned}
	\]
Hence $\hat U(t)=\hat U(0)$ for $t>0$. It follows that, for all $t>0$,
\[\begin{cases}
[\hat h(t)-\hat g(t)](\delta+\epsilon)\geq \hat U(t)=\hat U(0)\geq [\hat h(0)-\hat g(0)]\delta,\medskip\\
[\hat h(t)-\hat g(t)]\delta\leq \hat U(t)=\hat U(0)\leq [\hat h(0)-\hat g(0)](\delta+\epsilon).
\end{cases}
\]
We thus obtain, for all $t>0$,
\[
\frac\delta {\delta+\epsilon} [\hat h(0)-\hat g(0)]
\leq \hat h(t)-\hat g(t)\leq \frac{\delta+\epsilon}\delta [\hat h(0)-\hat g(0)].
\]
The monotonicity of $\hat g(t)$ and $\hat h(t)$ then implies
\[
 \hat h(t)\leq \frac{\delta+\epsilon}\delta [\hat h(0)-\hat g(0)]+\hat g(t)\leq \hat h(0)+\frac\epsilon\delta[\hat h(0)-\hat g(0)].
\]
Similarly
\[
\hat g(t)\geq \hat g(0)-\frac\epsilon\delta[\hat h(0)-\hat g(0)].
\]
By Step 3 we have 
\[
\hat h(0)-\hat g(0)=h(T)-g(T)+2\epsilon\leq 2(M+\epsilon).
\]
We thus obtain, for all $t>0$,
\[\begin{cases}
\hat h(t)\leq h(T)+\epsilon+\frac{2\epsilon}\delta (M+\epsilon),\\
\hat g(t)\geq g(T)-\epsilon-\frac{2\epsilon}\delta (M+\epsilon).
\end{cases}
\]

Since $\theta> \hat u(t,x)\geq \delta$ for all $t>0$ and $x\in [\hat g(t), \hat h(t)]$, it is easily checked that $(\hat u(t,x), \hat g(t),\hat h(t))$ and $(u(T+t,x), g(T+t),h(T+t))$
form a pair of upper and lower solutions as described in Lemma \ref{lemma2.3b}. It follows that
\[
[g(T+t), h(T+t)]\subset (\hat g(t), \hat h(t)) \mbox{ for } t>0.
\]
Therefore we have, for all $t>T$, 
\begin{equation}\label{T}\begin{cases}
 h(t)\leq h(T)+\epsilon+\frac{2\epsilon}\delta (M+\epsilon),\\
 g(t)\geq g(T)-\epsilon-\frac{2\epsilon}\delta (M+\epsilon).
\end{cases}
\end{equation}
This implies $h_\infty:=\lim_{t\to\infty} h(t)$ exists, for otherwise, there exists $t_n\to\infty$  such that
\[
h_*:=\liminf_{t\to\infty} h(t)=\lim_{n\to\infty} h(t_n)<h^*:=\limsup _{t\to\infty} h(t).
\]
For any small $\epsilon>0$ we may take $T=t_n$ with $n$ sufficiently large such that $h(T)\leq h_*+\epsilon$ and \eqref{T} holds.
Then for all $t>T$ we have
\[
h(t)\leq  h(T)+\epsilon+\frac{2\epsilon}\delta (M+\epsilon)\leq h_*+2\epsilon+\frac{2\epsilon}\delta (M+\epsilon).
\]
It follows that $h^*\leq h_*+2\epsilon+\frac{2\epsilon}\delta (M+\epsilon)$, which is impossible when $\epsilon$ is small enough.
We can similarly prove that $g_\infty:=\lim_{t\to\infty} g(t)$ exists.

 For the function $\displaystyle
	U(t) := \int_{g(t)}^{h(t)} u(t, x) \, dx$,
 similar to $\hat U(t)$ we have $U(t)=U(0)$ for $t>0$. It follows that
\[
\int_{-h_0}^{h_0}u_0(x)dx=U(0)=\lim_{t\to\infty} U(t)=\delta (h_\infty-g_\infty).
\]
The proof for the case $\delta\in (0, \theta)$ is now complete. 
\medskip

It remains to consider the case $\delta=\theta$. The proof for this case is a simple variation of the arguments in the  three steps above. To obtain the corresponding conclusions in Step 1, we simply observe that now
the constant triple $(\theta, -h_0-1, h_0+1)$ and $(u(t,x), g(t), h(t))$ form a pair of upper and lower solutions, which yields $u(t,x)\leq \theta=\delta$ and $[g(t), h(t)]\subset (-h_0-1, h_0+1)$
for all $t>0$. The corresponding conclusions in Step 2 and Step 3 are proved by the same arguments with considerable simplifications; for example,  we can now simply take $(\hat u, \hat g, \hat h)\equiv (\theta, g(T)-\epsilon, h(T)+\epsilon)$ in Step 3.
\hfill $\Box$

\subsubsection{Proof of Theorem \ref{th3.3}} The proof consists of four steps. The first three steps are based on  techniques already used  in this paper, but the fourth step relies on a new idea.
\medskip

\noindent
{\bf Step 1.} We prove that there exists $M>0$ such that $[g(t), h(t)]\subset (-M, M)$ for all $t>0$ and $\limsup _{t\to\infty} u(t,x)\leq 1$ uniformly for $x\in[g(t), h(t)]$.

We follow the approach in Step 1 of  subsection 3.6.1. 
 Choose $\bar u_0(x)\in X(h_0+1)$ such that 
\[
\bar u_0(x)\geq u_0(x) \mbox{ for } x\in [-h_0, h_0],\ \bar u_0(x)\geq 1 \mbox{ for } x\in [-h_0-1, h_0+1].
\]
Then consider the auxiliary free boundary problem
 \begin{equation*}\label{tilde-1}
  \begin{cases}
  \tilde u_t-d\tilde u_{xx}=0, & t>0, \ \tilde g(t)<x<\tilde h(t),\\
   \tilde u(t,\tilde g(t))= \tilde u(t,\tilde h(t))=1,&t>0,\\
    \tilde g'(t)=-d \tilde u_x(t, \tilde g(t)), & t>0,\\
    \tilde  h'(t)=-d \tilde u_x(t, \tilde h(t)),&t>0,\\
   \tilde  u(0,x)=\bar u_0(x),\ -\tilde g(0)=\tilde h(0)=h_0+1,
  \end{cases}
  \end{equation*}
  By the standard comparison principle we deduce $\tilde u(t,x)\geq 1$ for $t>0$ and $x\in [\tilde g(t), \tilde h(t)]$. Therefore $f(t,x,\tilde u)\leq 0$.
  We may now use Lemma \ref{lemma2.3b} to conclude that
  \begin{equation}\label{h-tilde h}
[g(t), h(t)]\subset (\tilde g(t), \tilde h(t)),\ u(t,x)\leq \tilde u(t,x) \mbox{ for } t>0,\ x\in [g(t), h(t)].
\end{equation}

  On the other hand, the triple $(\tilde u-1, \tilde g,\tilde h)$ solves a free boundary problem of the form \eqref{free-bound-0}, and as in Step 1 of  subsection 3.6.1, we can use \cite{dl2015} to conclude that only vanishing is possible for $(\tilde u-1, \tilde g,\tilde h)$. Hence
\[\begin{cases}
\lim_{t\to\infty}[\tilde g(t), \tilde h(t)]=[\tilde g_\infty, \tilde h_\infty] \mbox{ is a finite interval, } \\
\tilde u(t,x)-1\to 0 \mbox{ as } t\to\infty \mbox{ uniformly for } x\in [\tilde g(t),\tilde h(t)].
\end{cases}
\]
The desired conclusions clearly follow from this and \eqref{h-tilde h}.
\medskip

\noindent
{\bf Step 2.} We show that $u(t,x)\to 1$ as $t\to\infty$ uniformly for $x\in [g(t), h(t)]$.

Choose $\hat { f}(u)\leq \underline f(u)$ such that $\hat f$ satisfies ${\bf (f_b)}$. Then, by Lemma \ref{lemma3.1}, the following problem
	\begin{equation*}\label{tw}
		\left\{
		\begin{aligned}
			&dq'' - cq' + \hat{f}(q) = 0, \quad z \in \mathbb{R}, \\
			&q(-\infty) = 0, \quad q(\infty) = 1
		\end{aligned}
		\right.
	\end{equation*}
	has a solution pair $(c,q)=(c_0,q_0)$ with $c_0>0$ and $q_0(\cdot)$ strictly increasing.
	
	Next, we  make use of $q_0(z)$ to construct a lower solution to bound $u(t,x)$ from below as in the proof of Theorem \ref{th1.3}. Since $q_0(-\infty) = 0$, we can choose $L > 0$ sufficiently large such that $q_0(h_0 - L) \leq \min_{x \in [-h_0, h_0]} u_0(x)$. Define
	\[
	\underline{u}(t, x) := \max\{q_0(ct - x - L), q_0(ct + x - L)\}.
	\]
	Due to $\hat {f} (u)\leq f(t,x,u)$ we  see
	 that $\underline{u}(t, x)$ satisfies (in the weak sense)
	\[
	\underline u_t \leq  d \underline u_{xx} + f(t,x, \underline u)  \quad \text{for } t > 0, \; x \in \mathbb{R}.
	\]
	Moreover
	\[
		0 < \underline{u}(t, x) <1 \quad \text{for } x \in \mathbb{R},\ \mbox{which implies } \underline u(t, x)< u(t,x) \mbox{ for } t>0,\ x\in\{g(t), h(t)\},\]
		and
		\[
		\underline{u}(0, x) = \max\{q_0(-x - L), q_0(x - L)\} \leq q_0(h_0 - L) \leq u_0(x) \quad \text{for } x \in [-h_0, h_0].
	\]
	Therefore we can apply the standard comparison principle over $\{(t, x) : t > 0, x \in [g(t), h(t)]\}$ to deduce  $u(t, x) \geq \underline{u}(t, x)$ in this region. 
	Since
	\[
	\lim_{t \to \infty} \|\underline{u}(t,\cdot) - 1\|_{L^\infty(\mathbb{R})} = 0,
	\]
	we obtain
	\[
\liminf_{t\to\infty}{u}(t, x) \geq\lim_{t\to\infty}\underline{u}(t, x) =1 \quad \mbox{ uniformly for }  x \in [g(t), h(t)].
	\]
	The desired conclusion now follows from this and Step 1.
\medskip

\noindent
{\bf Step 3.} We show that, as $t \to \infty$,
	$
		[g(t), h(t)] \to [g_\infty, h_\infty] \mbox{ is a finite interval.}	
	$

	We follow  Step 3 of subsection 3.6.1.
	For any small $\epsilon>0$, we can find $T=T_\epsilon>0$ large such that
\[
u(t,x)\in [1-\epsilon, 1+\epsilon] \mbox{ for } t\geq T, \ x\in [g(t), h(t)].
\]

Next we consider the following auxiliary problem
\begin{equation}\label{hat-u-1}
 \begin{cases}
  \hat u_t-d\hat u_{xx}=0, & t>0, \ \hat g(t)<x<\hat h(t),\\
   \hat u(t,\hat g(t))= \hat u(t,\hat h(t))=1,&t>0,\\
    \hat g'(t)=-d\hat u_x(t, \hat g(t)), & t>0,\\
    \hat  h'(t)=-d \hat u_x(t, \hat h(t)),&t>0,\\
   \hat  u(0,x)=\hat u_0(x),\ \hat g(0)=g(T)-\epsilon, \hat h(0)=h(T)+\epsilon,& \hat g(0)\leq x\leq \hat h(0),
  \end{cases}
\end{equation}
where $\hat u_0$ is a $C^2$ function satisfying
\[
\hat u_0(\hat g(0))=\hat u_0(\hat h(0))=1,\ \max\{1, u(T, x)\} \leq \hat u_0(x)\leq 1+\epsilon \mbox{ for } x\in [\hat g(0), \hat h(0)].
\]
We easily see that the unique solution $(\hat u, \hat g, \hat h)$ of \eqref{hat-u-1} satisfies $1+\epsilon\geq \hat u(t,x)\geq 1$ for all $t>0$ and $x\in [\hat g(t), \hat h(t)]$.
It follows that $\hat h'(t)\geq 0\geq \hat g'(t)$ for $t>0$. As  in Step 3 of  subsection 3.6.1, the function
$\displaystyle
\hat U(t):=\int_{\hat g(t)}^{\hat h(t)} \hat u(t,x)dx$
satisfies
$\hat U(t)=\hat U(0)$ for $t>0$, which  leads to, for all $t>0$,
\[\begin{cases}
\hat h(t)\leq h(T)+\epsilon+2\epsilon (M+\epsilon),\medskip\\
\hat g(t)\geq g(T)-\epsilon-2\epsilon (M+\epsilon).
\end{cases}
\]

Since $\hat u(t,x)\geq 1$ for all $t>0$ and $x\in [\hat g(t), \hat h(t)]$, it is easily checked that $(\hat u(t,x), \hat g(t),\hat h(t))$ and $(u(T+t,x), g(T+t),h(T+t))$
form a pair of upper and lower solutions as described in Lemma \ref{lemma2.3b}. It follows that
\[
[g(T+t), h(T+t)]\subset (\hat g(t), \hat h(t)) \mbox{ for } t>0.
\]
Therefore we have, for all $t>T$, 
\begin{equation*}\label{T-1}\begin{cases}
 h(t)\leq h(T)+\epsilon+2\epsilon (M+\epsilon),\medskip\\
 g(t)\geq g(T)-\epsilon-2\epsilon (M+\epsilon).
\end{cases}
\end{equation*}
This implies, as in  Step 3 of  subsection 3.6.1, that $h_\infty:=\lim_{t\to\infty} h(t)$ and $g_\infty:=\lim_{t\to\infty} g(t)$ both exist.
\medskip

\noindent
{\bf Step 4.} We prove $h_\infty-g_\infty>0$.

As before, the function $\displaystyle
	U(t) := \int_{g(t)}^{h(t)} u(t, x) \, dx$ satisfies,
	 for $t >0$,
	\[\begin{aligned}
	U'(t) &=\int_{g(t)}^{h(t)} f(t,x,u(t,x)) \, dx \geq \int_{g(t)}^{h(t)} \underline f(u(t,x)) dx.
	 \end{aligned}
	\]

If  $u_0(x)\geq 1$ for $x\in [-h_0, h_0]$, then by a simple comparison consideration we deduce $u(t,x)\geq 1$ for $t>0$ and $x\in [g(t), h(t)]$, which implies $h'(t)\geq 0\geq g'(t)$ for $t>0$. Therefore $h_\infty\geq h(0)>g(0)\geq g_\infty$.

If   $u_0(x)\leq 1$ for $x\in [-h_0, h_0]$, then by a simple comparison consideration we deduce $ u(t,x)\leq 1$ for $t>0$ and $x\in [g(t), h(t)]$. Since $\underline f$ satisfies ${\bf (f_A)}$, we have $\underline f(u)>0$ for $u\in (\underline \theta, 1)$ with some fixed $\underline\theta\in [0, 1)$. By the conclusion of Step 2, we can find $T_0\geq 0$ such that
$1\geq u(t,x)>\underline\theta$ for $t\geq T_0$ and $x\in [g(t), h(t)]$. It follows that
\[
U'(t)\geq \int_{g(t)}^{h(t)} \underline f(u(t,x)) dx\geq 0 \mbox{ for } t\geq T_0,
\]
with $U'(t)\equiv 0$ for $t\geq T_0$ only if $u(t,x)\equiv 1$ for such $t$. In the case $u(t,x)\equiv 1$ for $t\geq T_0$, it follows that $h(t)\equiv h(T_0)=h_\infty>g_\infty=g(T_0)\equiv g(t)$ for $t>T_0$, and otherwise we have $U(\infty)>U(T_0)>0$ which implies $h_\infty>g_\infty$. 

If $u_0(x)-1$ changes sign in $[-h_0, h_0]$, then a new idea is required.  Now $m_0:=\max_{x\in[-h_0, h_0]}u_0(x)>1$. Let $v(t)$ be the unique solution of
\[
v'=\bar f(v),\ v(0)=m_0.
\]
Then from $\bar f(1)=0$ and $\bar f(u)<0$ for $u>1$ we see that $v(t)$ decreases to 1 as $t\to\infty$. Moreover, the usual comparison principle gives $u(t,x)\leq v(t)$ for $t>0$ and $x\in [g(t), h(t)]$. 

Since $\underline f'(1)<0$, there exists $\epsilon_1\in (0,\varepsilon]$ (where $\varepsilon$ is given in \eqref{additional}) such that $\underline f(u)$ is decreasing in $[1-\epsilon_1, 1+\epsilon_1]$.
 Using $u, v\to 1$ as $t\to\infty$, we can find $T>0$ large so that $u(t,x), v(t)$ both belong to $[1-\epsilon_1, 1+\epsilon_1]$ for $t\geq T$. Therefore, by the extra assumption \eqref{additional}, for $t\geq T$,
 \[
\begin{aligned}
 U'(t)&\geq  \int_{g(t)}^{h(t)} \underline f(u(t,x)) dx\geq \int_{g(t)}^{h(t)} \underline f(v(t)) dx\\
&\geq \varepsilon^{-1} [h(t)-g(t)]\bar f(v(t))=\varepsilon^{-1} [h(t)-g(t)]v'(t),
\end{aligned}
\]
and so, for any $t>s\geq T$,
\[
\int_{g(t)}^{h(t)}u(t,x)dx-\int_{g(s)}^{h(s)}u(s,x)dx=U(t)-U(s)\geq \varepsilon^{-1}\int_s^t[h(\tau)-g(\tau)]v'(\tau)d\tau.
\]

Assuming $h_\infty=g_\infty$, we now deduce a contradiction. Using $0<h(t)-g(t)\to 0$ as $t\to\infty$, we can find a positive sequence $t_n\to\infty$ such that
$h(t_n)-g(t_n)\geq h(t)-g(t)$ for $t\geq t_n$. Therefore for all large $n$ and $t>t_n$, in view of $v'<0$, we have
\[
\int_{t_n}^t[h(\tau)-g(\tau)]v'(\tau)d\tau\geq [h(t_n)-g(t_n)]\int_{t_n}^tv'(\tau)d\tau\geq  [h(t_n)-g(t_n)](1-v(t_n)).
\]
It follows that, for all large $n$ and $t>t_n$,
\begin{equation}\label{t-tn}\begin{aligned}
\int_{g(t)}^{h(t)}u(t,x)dx&\geq \int_{g(t_n)}^{h(t_n)}u(t_n,x)dx+\varepsilon^{-1}[h(t_n)-g(t_n)](1-v(t_n))\\
&=\int_{g(t_n)}^{h(t_n)}[u(t_n,x)+\varepsilon^{-1}(1-v(t_n))]dx.
\end{aligned}
\end{equation}
Since $u(t,x), v(t)\to 1$ as $t\to\infty$ uniformly in $x\in [g(t), h(t)]$, we can fix $n$ sufficiently large such that 
\[
u(t_n,x)+\varepsilon^{-1}(1-v(t_n))>0 \mbox{ for } x\in [g(t_n), h(t_n)],\]
 which implies 
 \[
  I_n:=\int_{g(t_n)}^{h(t_n)}[u(t_n,x)dx+\varepsilon^{-1}(1-v(t_n))]dx>0.
\]
On the other hand, $h_\infty=g_\infty$ and $u(t,x)\to 1$ as $t\to\infty$ imply $\displaystyle\int_{g(t)}^{h(t)}u(t,x)dx\to 0$ as $t\to\infty$.
Thus letting $t\to\infty$ in \eqref{t-tn} we deduce $0\geq I_n>0$. This contradiction proves $h_\infty>g_\infty$.
The proof is now complete.
\hfill $\Box$

\bigskip

\bigskip

\paragraph{\bf Acknowledgements.} The research of Y. Du and W. Ni was supported by the Australian Research Council; part of this work was carried out while L. Li was visiting the University of New England as a visiting PhD student; N. Shabgard was supported by a targeted PhD scholarship of the University of New England. We thank Professor Tim Scharef for useful discussions on numerical simulations in the paper.

 \end{document}